\documentclass[11pt]{article}

\usepackage{amsmath,amssymb,amsthm}
\usepackage{mathtools}
\usepackage{bm}
\usepackage{geometry}
\usepackage{graphicx}
\usepackage{hyperref}
\usepackage{natbib}
\usepackage{booktabs}
\usepackage{placeins}
\allowdisplaybreaks

\title{Fast Stochastic Nearest--Neighbor Pairwise Composite Likelihood for Massive Spatial Datasets}

\author{
Moreno Bevilacqua\textsuperscript{1,2}\thanks{Corresponding author. Email: \texttt{moreno.bevilacqua@uai.cl}},
Francisco Cuevas-Pacheco\textsuperscript{3},
Christian Caama\~no-Carrillo\textsuperscript{4}
\\[0.5em]
\textsuperscript{1}Universidad Adolfo Ib\'a\~nez, Chile\\
\textsuperscript{2}Ca' Foscari University of Venice, Italy\\
\textsuperscript{3}Universidad T\'ecnica Federico Santa Mar\'ia, Chile\\
\textsuperscript{4}Universidad del B\'io-B\'io, Chile
}

\date{}
\newtheorem{proposition}{Proposition}

\newcommand{\R}{\mathbb{R}}
\newcommand{\EE}{\mathbb{E}}
\newcommand{\PP}{\mathbb{P}}
\newcommand{\Var}{\mathrm{Var}}
\newcommand{\Cov}{\mathrm{Cov}}

\newcommand{\argmax}{\mathop{\mathrm{argmax}}}

\begin{document}
\maketitle

\begin{abstract}
Weighted pairwise composite likelihoods based on nearest-neighbor (NN) pairs
provide a scalable alternative to full likelihood inference for spatial random
fields, but can remain expensive when moderately large NN neighborhoods are
needed. We propose a stochastic acceleration that constructs the deterministic
NN candidate graph and evaluates only a randomized subset of its pairwise
contributions.

We consider two thinning designs: Bernoulli thinning, which controls the
retained-pair budget in expectation, and fixed-budget thinning, which enforces
an exact budget through target-wise sampling without replacement. Simulation
studies for Mat\'ern covariance models suggest that, in the settings
considered, retaining two pairs per observation provides a stable
statistical--computational compromise. The stochastic NN pairwise estimators
provide a faster alternative to a Vecchia-type Gaussian approximation when
substantial reductions in covariance-fitting time are desired and a modest loss
of efficiency is acceptable. This trade-off is especially favorable for the mean, scale, and
sill parameters, while the main efficiency loss is concentrated on smoothness
estimation.

In an application to July average temperature over the western--central United
States, based on 2.5 million WorldClim observations, the proposed estimators
achieve predictive accuracy essentially indistinguishable from the Vecchia
benchmark, with substantially shorter covariance-fitting time.
\end{abstract}

\section{Introduction}

Gaussian random fields are a fundamental tool for modelling spatial and
spatio-temporal dependence in environmental sciences, climate studies,
geostatistics, epidemiology, and remote sensing. Likelihood-based inference is
particularly attractive because it allows joint estimation of mean and
covariance parameters, including variance, range, and smoothness. However,
modern datasets often contain tens of thousands to millions of observations.
For Gaussian random fields, exact likelihood evaluation requires
\(\mathcal{O}(n^3)\) operations and \(\mathcal{O}(n^2)\) memory storage, making
full maximum likelihood inference infeasible even when \(n\) is only moderately
large.

This computational bottleneck has motivated a large literature on scalable
methods for likelihood-based inference with spatial Gaussian random fields.
Important approaches include low-rank covariance representations
\citep{Banerjee:2008,cress+j08}, covariance tapering
\citep{fur+g+n06,kauff}, Gaussian Markov random field and SPDE
approximations \citep{llinj}, multiresolution approximations
\citep{Katzfuss2017,katzfuss2020class}, and Vecchia-type approximations
\citep{vecchia,guigui,STS755}. For a broad review and comparison of scalable
methods for massive spatial datasets, see \citet{Heatonetal:2019}.

Composite likelihood methods \citep{lindsay88,varin+r+f11} provide a widely
used alternative to full likelihood inference by replacing the full likelihood
with a product of low-dimensional marginal or conditional likelihood
contributions. In spatial statistics, weighted pairwise composite likelihoods
\citep{Bevilacqua+g+m+p12,Bevilacqua:Gaetan:2015} are appealing because they
retain information on spatial dependence while avoiding the computational
burden of the full joint likelihood. They are also useful for complex
non-Gaussian random fields, for which the full finite-dimensional distribution
may be unavailable or computationally intractable, whereas bivariate
distributions can often be specified or evaluated
\citep{Heagerty:Lele:1998,Bevilacqua_et_al:2021,poi22}. Nevertheless, without
further restriction, pairwise composite likelihoods involve \(\mathcal{O}(n^2)\)
pairs and therefore remain expensive for massive datasets.

A statistically and computationally convenient way to reduce this cost is to
use weights based on a nearest-neighbor (NN) graph. In the NN weighted pairwise
likelihood of \citet{NNP}, only pairs associated with the \(m\) nearest
neighbors of each location are retained. For fixed \(m\), this reduces the
cost of evaluating the composite-likelihood objective function to
\(\mathcal{O}(nm)\) and the memory requirement to \(\mathcal{O}(n)\). This
deterministic NN weighted pairwise likelihood is the starting point of the
present paper. However, for massive datasets, the number of bivariate
likelihood evaluations can still be substantial when moderate or large values
of \(m\) are needed. This situation may arise, for example, when estimating
the smoothness parameter of flexible correlation models, such as the Mat\'ern
\citep{porcumatern} or Generalized Wendland \citep{ana2019} models.

The main goal of this paper is to reduce the computational cost of NN weighted
pairwise composite likelihoods. We propose a stochastic acceleration that first
constructs the deterministic NN candidate graph and then evaluates only a
randomized subset of its pairwise contributions. This separates two roles that
are coupled in the deterministic NN pairwise likelihood: the number of nearest
neighbors \(m\) controls the richness of the local candidate graph, whereas the
thinning parameter \(p\in(0,1]\) controls the retained-pair budget. Conditional
on the precomputed NN lists, the number of bivariate likelihood evaluations is
reduced from \(\mathcal{O}(nm)\) to approximately \(\mathcal{O}(pnm)\). Thus,
one can use richer NN candidate graphs while keeping the number of evaluated
pairwise likelihood terms computationally affordable.

The idea is related to recent work on stochastic composite likelihoods
\citep{Mazo:2024,Alfonzetti:2025}, where only a random subset of composite
contributions is evaluated in order to control computational cost. While those
contributions focus on settings with independent replicates, the present paper
considers a single realization of a spatial random field, where the pairwise
likelihood contributions are spatially dependent. In addition, the
randomization is performed within a deterministic spatial NN graph, rather than
over an unstructured collection of composite contributions, leading to a
randomized version of a local spatial estimating criterion.

The contribution of this paper is threefold. First, we formulate stochastic NN
pairwise likelihood as a two-stage procedure that separates the construction of
a deterministic NN candidate graph from the randomized selection of the pairs
actually evaluated. The neighborhood size \(m\) controls the richness of the
candidate graph, whereas the thinning parameter controls the computational
budget. Second, we study two thinning designs on this graph: calibrated
independent Bernoulli thinning, which controls the retained-pair budget in
expectation, and fixed-budget thinning, which fixes the total number of
evaluated pairs exactly through target-wise sampling without replacement.
Third, we provide a systematic empirical study, up to spatial datasets with
millions of observations, of the trade-off between retained-pair budget,
computing time, and loss of efficiency for each Mat\'ern parameter, using a
Vecchia-type approximation as an external benchmark.

The simulation studies suggest that retaining two pairs per observation
provides a stable statistical--computational compromise in the Mat\'ern
settings considered. This value should be interpreted as an empirical
calibration for the present simulation designs, rather than as a universal
recommendation. At this budget, the stochastic NN pairwise estimators remain
close to the deterministic FullNN estimator for the mean, scale, and sill
parameters, while the main loss of efficiency is concentrated on the Mat\'ern
smoothness parameter. Compared with the Vecchia-type benchmark, the proposed
estimators are substantially faster and retain comparable accuracy for the
mean, scale, and sill parameters, whereas the main efficiency loss is again
concentrated on smoothness estimation, especially in the smoother Mat\'ern
scenarios. Additional simulations show that the internal Monte Carlo
variability due to thinning is negligible for the mean and sill, moderate for
the scale, and largest for smoothness.

The two thinning designs provide different forms of budget control. Bernoulli
thinning is simpler and controls the retained size only in expectation, whereas
fixed-budget thinning enforces the retained-pair budget exactly and allocates it
across target-specific NN lists. The simulations show that this additional
control reduces the internal thinning variability for the mean and sill
parameters, but does not imply uniform RMSE dominance over Bernoulli thinning.

The large-scale temperature application further shows that the proposed
stochastic NN pairwise fits can be used on multi-million-observation datasets
and can achieve predictive accuracy essentially indistinguishable from a
Vecchia benchmark, with substantially shorter covariance-fitting times.

The inferential properties of the resulting estimators are discussed under
increasing-domain weak-dependence conditions. Since direct estimation of the
Godambe variance can be difficult, we use a parametric score bootstrap to
compute standard errors. The computational savings from thinning, together with
fast simulation methods (see \citealp{newfast} and the references therein),
make simulation-based estimation of the Godambe variance feasible for massive
spatial datasets. The proposed methods are implemented in the \texttt{GeoModels}
package for \textsf{R} \citep{GeoModels}.

Although the paper focuses on Gaussian spatial random fields, the proposed
randomization is defined at the level of the candidate pair graph and of the
associated pairwise likelihood contributions. The same construction can
therefore be adapted to other pairwise composite likelihood settings, including
non-Gaussian, space--time, and multivariate random fields, whenever the relevant
bivariate likelihood contributions can be evaluated.

The remainder of the paper is organized as follows. Section~\ref{sec:wpl}
reviews NN weighted pairwise composite likelihood. Section~\ref{sec:two_designs}
introduces the two stochastic thinning designs. Section~\ref{sec:inference}
discusses standard errors and the parametric bootstrap.
Sections~\ref{sec:simulation} and~\ref{sec:application} present simulation and
data examples. Section~\ref{sec:conclusion} concludes.
Appendix~\ref{app:asymptotics} gives asymptotic statements and discusses
Godambe variability.

\section{Weighted composite likelihood based on nearest--neighbor pairs}
\label{sec:wpl}

Let
\[
Z=\{Z(\mathbf{s}),\mathbf{s}\in A\subset\R^d\}
\]
be a spatial random field observed at locations
\(\mathbf{s}_1,\ldots,\mathbf{s}_n\). The distribution of the field is indexed
by a parameter vector \(\bm\theta\in\Theta\subset\R^q\), which may include mean
parameters, marginal variance parameters, nugget parameters, and parameters
governing the spatial correlation function.

Given two observation sites \(\mathbf{s}_i\) and \(\mathbf{s}_j\), denote by
\[
\mathbf{Z}_{ij}=\{Z(\mathbf{s}_i),Z(\mathbf{s}_j)\}^{\top}
\]
the corresponding bivariate random vector, with joint density
\(f_{ij}(\mathbf{z}_{ij};\bm\theta)\). The weighted pairwise composite log-likelihood can be based either on
bivariate marginal contributions or on pairwise conditional contributions. In
this paper we focus on the bivariate marginal version, which is the one used in
the simulation studies and in the application. We therefore write
\[
\ell(\mathbf Z_{ij};\bm\theta)
=
\log f_{ij}(\mathbf Z_{ij};\bm\theta)
\]
and define
\begin{equation}
\label{eq:wcl}
wpl(\bm\theta)=
\sum_{i=1}^n\sum_{j\neq i}
\ell(\mathbf{Z}_{ij};\bm\theta) w_{ij}.
\end{equation}
The role of the non-negative weights \(w_{ij}\) is to save computational time
and improve the statistical efficiency. Compactly supported weights have been
shown to accomplish this task \citep{Bevilacqua:Gaetan:2015,dddavis}.

\cite{NNP} proposed compactly supported weights based on NN.
Let \(N_m(\mathbf{s}_j)\) denote the set of the \(m\) nearest neighbors of
\(\mathbf{s}_j\). Directed nearest-neighbor weights are defined as
\begin{equation}
\label{eq:nnweights}
w^{NN}_{ij}(m)=
\begin{cases}
1, & \mathbf{s}_i\in N_m(\mathbf{s}_j),\\
0, & \text{otherwise},
\end{cases}
\qquad i,j\in\{1,\ldots,n\}.
\end{equation}
Here \(j\) is the target location and \(i\) is one of its nearest neighbors.
The NN weights are directed and, in general, non-symmetric. Therefore, an
unordered pair \(\{i,j\}\) may enter the directed NN composite likelihood once,
twice, or not at all.

This NN construction is computationally convenient because kd-tree algorithms,
such as those of \citet{Arya_etal1998}, allow efficient construction of
nearest-neighbor sets. Once the NN lists are available, evaluation and
maximization of \(wpl(\bm\theta)\) can be performed by summing only over NN
pairs. Thus, for fixed \(m\), the deterministic NN weighted composite
likelihood requires \(\mathcal{O}(nm)\) bivariate likelihood evaluations and
\(\mathcal{O}(n)\) memory storage \citep{NNP}.

The number of nearest neighbors \(m\) controls both the computational cost and
the amount of pairwise dependence information entering the estimating
criterion. Increasing \(m\) enlarges the set of non-zero NN weights and can add
information from moderate spatial lags. However, efficiency is not necessarily
monotone in \(m\): additional pairs may be weakly informative, redundant with
nearer pairs, or strongly dependent on other pairwise contributions. When the
smoothness parameter is fixed, relatively small values of \(m\) may be
adequate. When smoothness is estimated jointly with range and variance, larger
NN neighborhoods can be beneficial because they include pairwise information
over a richer range of spatial lags \citep{NNP}.

The stochastic methods proposed below keep the deterministic NN candidate set
as the starting point, but evaluate only a random subset of its pairwise
contributions. Thus, \(m\) determines the richness of the candidate NN graph,
whereas the thinning mechanism determines how many of those candidate pairs are
used in the composite likelihood.

\section{Two stochastic thinning strategies}
\label{sec:two_designs}

Let
\[
\mathcal{E}_n(m)=\{(i,j): w^{NN}_{ij}(m)=1,\; j\neq i\}
\]
denote the set of candidate directed NN pairs, with size
\[
d_n:=|\mathcal{E}_n(m)|=\mathcal{O}(nm).
\]
In the directed convention used here, \(j\) is the target location and \(i\) is
one of its nearest neighbors.

We define randomized NN weights by
\begin{equation}
\label{eq:thinweights_general}
\widetilde w_{ij}=w^{NN}_{ij}(m)X_{ij},
\qquad i,j\in\{1,\ldots,n\},\quad i\neq j,
\end{equation}
where \(X_{ij}\in\{0,1\}\) are random inclusion indicators generated
independently of the observed field values. Since \(w^{NN}_{ij}(m)=0\) outside
the NN candidate set, \(\widetilde w_{ij}=0\) for all non-candidate pairs.
The thinned weighted pairwise composite log-likelihood is
\begin{equation}
\label{eq:thin_wcl}
\widetilde{wpl}(\bm\theta)=
\sum_{i=1}^n\sum_{j\neq i}
\widetilde w_{ij}\ell(\mathbf{Z}_{ij};\bm\theta)
=
\sum_{(i,j)\in\mathcal{E}_n(m)}
X_{ij}\ell(\mathbf{Z}_{ij};\bm\theta).
\end{equation}
Multiplying \eqref{eq:thin_wcl} by a deterministic positive constant, such as
the reciprocal of the retained size or expected retained size, does not change
the maximizer. Normalization is nevertheless useful in the asymptotic
arguments of Appendix~\ref{app:asymptotics}.

The parameter \(p\in(0,1]\) controls the computational budget in both designs.
In Method I, \(p\) is the marginal inclusion probability in the
constant-probability Bernoulli case. In Method II, \(p\) is the nominal retained
fraction and defines the fixed retained-pair budget
\(K_{\mathrm{tar}}=\operatorname{round}(p d_n)\). Thus, comparisons between the
two methods should be made at comparable retained-pair budgets.

Figure~\ref{fig:toy_thinning} illustrates the two stochastic thinning
mechanisms using two reproducible toy examples. In each example, \(n\)
locations are generated independently and uniformly on the unit square
\([0,1]^2\). The full pairwise Euclidean distance matrix is then computed, and,
for each target location \(j\), the \(m\) nearest neighbors \(i\) are retained
to form the directed \(m\)-NN candidate graph. Hence, in these regular toy
examples, the candidate graph contains \(d_n=nm\) directed edges. Starting from
this common candidate graph, Method I retains each candidate edge independently
with probability \(p\), whereas Method II sets
\(K_{\mathrm{tar}}=\operatorname{round}(p d_n)\) and retains exactly
\(K_{\mathrm{tar}}\) edges by allocating the budget across target-specific
nearest-neighbor lists and sampling without replacement within each list.

\begin{figure}[!htbp]
\centering

\includegraphics[width=\textwidth]{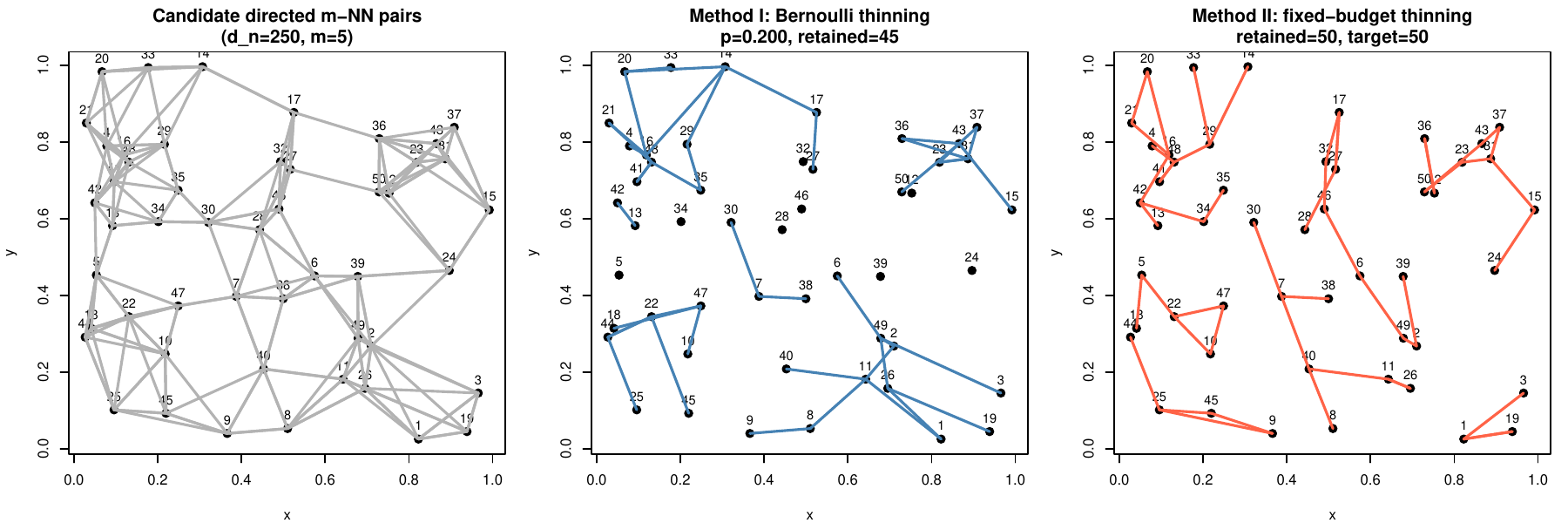}

\vspace{0.5em}

\includegraphics[width=\textwidth]{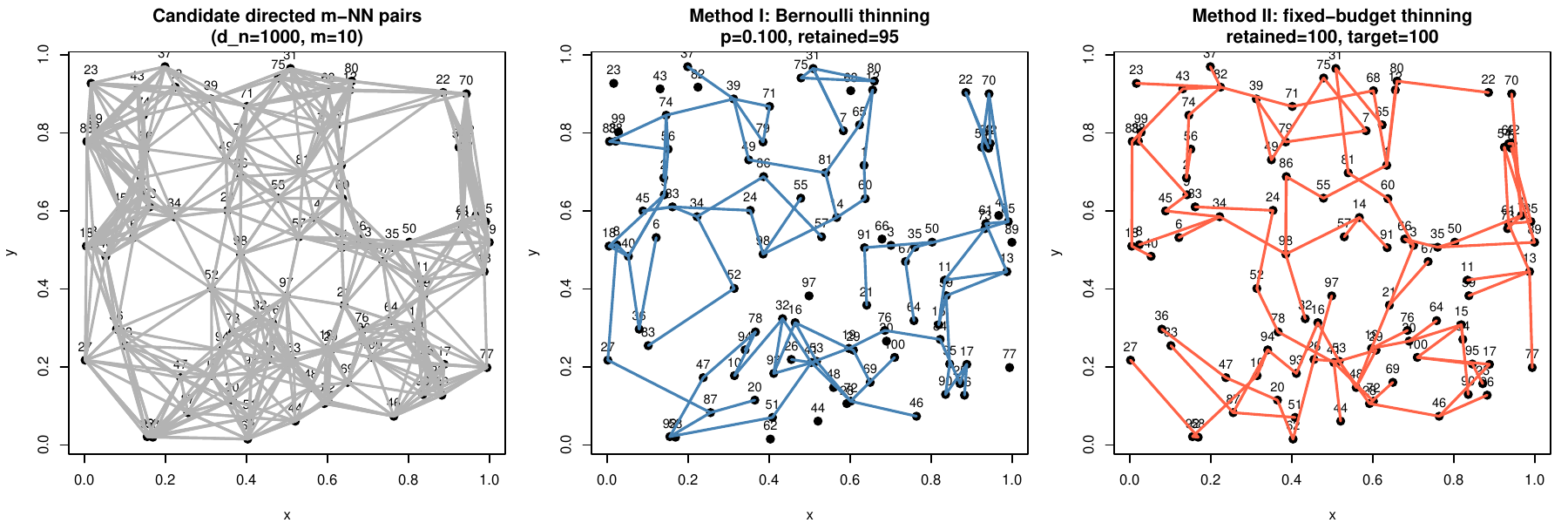}

\caption{
Toy illustrations of the two stochastic thinning designs on directed
nearest-neighbor candidate graphs. In each row, the left panel shows the full
directed \(m\)-NN candidate graph, the middle panel shows Method I based on
independent Bernoulli thinning, and the right panel shows Method II based on
fixed-budget target-wise sampling without replacement. The point locations are
generated independently and uniformly on \([0,1]^2\). For each target location
\(j\), the candidate graph is built by retaining its \(m\) nearest neighbors
\(i\), yielding \(d_n=nm\) directed candidate pairs. Method I retains each
candidate pair independently with probability \(p\), whereas Method II uses the
same nominal fraction \(p\) to define the exact budget
\(K_{\mathrm{tar}}=\operatorname{round}(p d_n)\), which is allocated across
target-specific nearest-neighbor lists. The top row uses \(n=50\), \(m=5\), and
\(p=0.2\), giving \(d_n=250\), 45 Bernoulli-retained pairs, and 50
fixed-budget pairs. The bottom row uses \(n=100\), \(m=10\), and \(p=0.1\),
giving \(d_n=1000\), 95 Bernoulli-retained pairs, and 100 fixed-budget pairs.
}
\label{fig:toy_thinning}
\end{figure}

\FloatBarrier

\subsection{Method I: calibrated independent Bernoulli thinning}

In the calibrated independent Bernoulli design, the indicators are mutually
independent conditional on the candidate set:
\[
X_{ij}\stackrel{\mathrm{ind}}{\sim}\mathrm{Bernoulli}(\pi_{ij}),
\qquad (i,j)\in\mathcal{E}_n(m).
\]
The simplest version uses constant inclusion probabilities,
\[
\pi_{ij}=p,
\qquad (i,j)\in\mathcal{E}_n(m),
\]
so that
\[
\EE_X\sum_{(i,j)\in\mathcal{E}_n(m)}X_{ij}=p d_n.
\]
Thus \(p\) directly controls the expected computational budget.

More generally, pair-specific probabilities can be used, provided they are
calibrated to satisfy
\begin{equation}
\label{eq:budget_expect}
\sum_{(i,j)\in\mathcal{E}_n(m)}\pi_{ij}=p d_n.
\end{equation}
For instance,
\begin{equation}
\label{eq:pi_weighted_general}
\pi_{ij}=\min\{1,c\omega_{ij}\},
\qquad c>0,
\end{equation}
where \(\omega_{ij}>0\) is a deterministic weight depending only on design
features of the pair, and \(c\) is chosen so that \eqref{eq:budget_expect}
holds. Without inverse-probability reweighting, pair-specific probabilities
define a weighted composite likelihood, with weights proportional to
\(\pi_{ij}\). In the simulations reported below, Method I uses the
constant-weight version \(\omega_{ij}\equiv 1\), so that \(\pi_{ij}=p\).

\subsection{Method II: fixed-budget stochastic thinning}
\label{sec:fixed_budget}

Method II is the fixed-budget design. It retains an exact number of candidate
NN pairs and allocates this budget across the target-specific NN lists. The
design depends only on the NN candidate set; it does not depend on the observed
field values or on unknown covariance parameters.

For each target location \(\mathbf s_j\), define the target-specific candidate
list
\[
\mathcal E_{n,j}(m)
=
\{(i,j): \mathbf s_i\in N_m(\mathbf s_j)\},
\qquad
d_{n,j}=|\mathcal E_{n,j}(m)|.
\]
Then
\[
\mathcal E_n(m)=\bigcup_{j=1}^n \mathcal E_{n,j}(m),
\qquad
d_n=\sum_{j=1}^n d_{n,j}.
\]
Given \(p\in(0,1]\), define the global retained-pair budget
\begin{equation}
\label{eq:fixed_budget}
K_{\mathrm{tar}}=\operatorname{round}(p d_n),
\qquad 0\le K_{\mathrm{tar}}\le d_n.
\end{equation}

The fixed-budget design assigns a local budget \(k_j\) to each target-specific
list. The local budgets are proportional to \(p d_{n,j}\), up to randomized
rounding, and satisfy the exact global constraint
\[
\sum_{j=1}^n k_j = K_{\mathrm{tar}}.
\]
In the implementation, this is obtained by taking deterministic floors
\(\lfloor p d_{n,j}\rfloor\) and allocating the remaining budget according to
the fractional parts. Thus, no tuning parameter beyond \(p\), or equivalently
\(K_{\mathrm{tar}}\), is introduced.

Conditional on the local budgets \(k_1,\ldots,k_n\), the retained set
\(\mathcal F_n(m)\) is obtained by sampling exactly \(k_j\) edges uniformly
without replacement from each target-specific list \(\mathcal E_{n,j}(m)\):
\[
\mathcal F_{n,j}(m)\subseteq \mathcal E_{n,j}(m),
\qquad
|\mathcal F_{n,j}(m)|=k_j,
\]
and
\[
\mathcal F_n(m)=\bigcup_{j=1}^n \mathcal F_{n,j}(m),
\qquad
|\mathcal F_n(m)|=K_{\mathrm{tar}}.
\]
The corresponding indicators are
\[
X_{ij}=\mathbf 1\{(i,j)\in\mathcal F_{n,j}(m)\},
\qquad
\widetilde w_{ij}=w^{NN}_{ij}(m)X_{ij}.
\]

Thus, Method II can be interpreted as proportional stratified sampling on the
directed NN graph, where the strata are the target-specific NN lists
\(\mathcal E_{n,j}(m)\). The total sample size is fixed at
\(K_{\mathrm{tar}}\), and sampling is performed without replacement within each
stratum. This connects the fixed-budget design to classical finite-population
sampling ideas, in particular stratified sampling and fixed-size sampling
without replacement \citep{Cochran:1977,Tille:2006}.

Conditional on the local budget \(k_j\), any edge \(e=(i,j)\in\mathcal
E_{n,j}(m)\) has inclusion probability
\[
\PP_X(X_e=1\mid k_j)=\frac{k_j}{d_{n,j}}.
\]
If the directed NN graph is regular, so that \(d_{n,j}=m\) for all \(j\), and
if \(pm\) is an integer, then \(k_j=pm\) for every target and the marginal
inclusion probability is exactly \(p\). For two distinct candidate edges
\(e\ne f\) in the same target-specific list \(\mathcal E_{n,j}(m)\),
\[
\PP_X(X_e=1,X_f=1\mid k_j)
=
\frac{k_j(k_j-1)}{d_{n,j}(d_{n,j}-1)},
\]
and therefore
\begin{equation}
\label{eq:fixed_budget_cov_target}
\Cov_X(X_e,X_f\mid k_j)
=
-\frac{k_j(d_{n,j}-k_j)}
{d_{n,j}^2(d_{n,j}-1)}
\le 0.
\end{equation}
Thus, the fixed-budget design induces negative dependence locally among
pairwise terms sharing the same target.

The difference between the two designs is therefore not only the randomness of
the retained size. Bernoulli thinning has independent indicators and a random
retained size with mean \(p d_n\). Fixed-budget thinning has an exact retained
size \(K_{\mathrm{tar}}\) and imposes a target-wise without-replacement
constraint. This provides exact computational control, avoids over- or
under-sampling individual target lists, and reduces redundant sampling within
each local neighborhood. This distinction is also visible in
Figure~\ref{fig:toy_thinning}: the Bernoulli design may retain a number of
edges different from its expectation, whereas the fixed-budget design matches
the target retained size exactly.

\subsection{Algorithmic summary and computational cost}
\label{sec:algorithm_cost}

The proposed estimator can be viewed as a two-stage procedure: first construct
a deterministic directed NN candidate graph, and then evaluate only a
randomized subset of its pairwise likelihood contributions.

\begin{enumerate}
\item Construct the directed \(m\)-NN candidate graph
\[
\mathcal{E}_n(m)=\{(i,j):\mathbf{s}_i\in N_m(\mathbf{s}_j),\; i\neq j\},
\qquad d_n=|\mathcal{E}_n(m)|.
\]
This step depends only on the sampling locations. Using kd-tree type nearest
neighbor search algorithms, the tree construction typically requires
\(\mathcal{O}(n\log n)\) operations and \(\mathcal{O}(n)\) storage
\citep{Arya_etal1998,NNP}. In the directed \(m\)-NN case,
\(d_n\simeq nm\).

\item Specify the nominal retained fraction \(p\in(0,1]\). For a given
candidate graph size \(m\), \(p\) controls the thinning intensity and determines
the target retained-pair budget
\[
K_{\mathrm{tar}}=\operatorname{round}(p d_n).
\]
For Bernoulli thinning, \(p\) is the marginal inclusion probability and the
retained size has expectation \(p d_n\). For fixed-budget thinning,
\(K_{\mathrm{tar}}\) is the exact number of retained pairs; allocating this
budget across the \(n\) target-specific NN lists requires one pass over the
lists and has cost \(\mathcal{O}(n)\).

In the empirical sections we often report the equivalent budget per
observation, \(K_{\mathrm{tar}}/n\), because it is directly interpretable as
the number of retained pairwise contributions per observation. For a regular
directed \(m\)-NN graph,
\[
p \simeq \frac{K_{\mathrm{tar}}/n}{m}.
\]

\item Generate the retained edge set. Under Bernoulli thinning, each candidate
edge is retained independently with probability \(p\). Under fixed-budget
thinning, sampling is performed without replacement within the target-specific
NN lists, using the allocated local budgets. Conditional on the NN graph,
generating the retained edge set requires scanning the candidate lists and has
cost \(\mathcal{O}(d_n)\).

\item Maximize the thinned pairwise composite log-likelihood
\[
\widetilde{wpl}(\bm\theta)
=
\sum_{(i,j)\in\mathcal{E}_n(m)}
X_{ij}\ell(\mathbf{Z}_{ij};\bm\theta),
\]
to obtain
\[
\widehat{\bm\theta}
=
\argmax_{\bm\theta\in\Theta}
\widetilde{wpl}(\bm\theta).
\]
Each objective evaluation involves only the retained pairs and therefore has
cost \(\mathcal{O}(K_n)\), where \(K_n\) denotes the realized retained size.

\item If standard errors are required, estimate the Godambe covariance matrix
using the parametric score bootstrap described in Section~\ref{sec:inference},
reusing the same NN candidate graph and applying the same thinning design and
retained-pair budget to each simulated dataset.
\end{enumerate}

Thus, for fixed \(m\), the deterministic NN pairwise likelihood evaluates
approximately \(d_n\simeq nm\) bivariate likelihood terms per objective
evaluation. The stochastic versions reduce this number to
\[
K_n \approx p d_n \simeq pnm
\]
under Bernoulli thinning, in expectation, and to
\[
K_n=K_{\mathrm{tar}}=\operatorname{round}(p d_n)
\]
under fixed-budget thinning. If \(K_n=\mathcal{O}(n)\), each stochastic
objective evaluation is linear in \(n\), conditional on the one-time
construction of the NN candidate graph. The parameter \(m\) controls the
richness of the candidate graph, whereas \(p\), for a given \(m\), controls the
effective number of pairwise likelihood terms actually evaluated.

\section{Inference and standard errors}
\label{sec:inference}

Composite-likelihood estimators are commonly accompanied by sandwich, or
Godambe, standard errors. For a bivariate marginal pairwise contribution
\(\ell(\mathbf Z_{ij};\bm\theta)\), define the score contribution and the
negative Hessian as
\[
\bm{s}_{ij}(\bm\theta)
=
\nabla_{\bm\theta}\ell(\mathbf Z_{ij};\bm\theta),
\qquad
\bm{h}_{ij}(\bm\theta)
=
-\nabla^2_{\bm\theta}\ell(\mathbf Z_{ij};\bm\theta).
\]
For a realized thinned edge set
\[
\mathcal S_n
=
\{(i,j)\in\mathcal E_n(m):\widetilde w_{ij}=1\},
\qquad
K_n=|\mathcal S_n|,
\]
the corresponding composite score and sensitivity matrix are
\[
\widetilde{\bm U}_{n}(\bm\theta)
=
\sum_{(i,j)\in\mathcal S_n}
\bm{s}_{ij}(\bm\theta),
\]
and
\[
\widetilde{\bm H}_{n}(\bm\theta)
=
\sum_{(i,j)\in\mathcal S_n}
\bm{h}_{ij}(\bm\theta).
\]

Under the increasing-domain weak-dependence conditions commonly used for local
weighted pairwise composite likelihoods, such estimators are consistent and
asymptotically normal, with asymptotic covariance given by the inverse Godambe
information; see, for example, \citet{Bevilacqua:Gaetan:2015}. In the present
setting, once the randomized thinning weights are realized, the objective is
again a local weighted pairwise composite likelihood, but based on the selected
edge set \(\mathcal S_n\). Appendix~\ref{app:asymptotics} states the
corresponding asymptotic properties under Bernoulli and fixed-budget thinning.

We use an unconditional-design interpretation of the stochastic estimator, in
which randomness comes from both the spatial field and the thinning design. Let
\[
\bm H
=
\lim_{n\to\infty}
K_n^{-1}
\EE\{\widetilde{\bm H}_{n}(\bm\theta_0)\},
\qquad
\bm J
=
\lim_{n\to\infty}
K_n^{-1}
\Var\{\widetilde{\bm U}_{n}(\bm\theta_0)\},
\]
where the expectation and variance are taken over both sources of randomness.
Then
\[
K_n^{1/2}(\widehat{\bm\theta}-\bm\theta_0)
\Rightarrow
N\{0,\bm H^{-1}\bm J\bm H^{-1}\}.
\]
Thus, up to the \(K_n^{-1}\) scaling already incorporated in the definitions of
\(\bm H\) and \(\bm J\), the asymptotic covariance has the usual sandwich form
\[
\bm G^{-1}
=
\bm H^{-1}\bm J\bm H^{-1}.
\]

Direct estimation of \(\bm J\) is difficult in spatial applications because
the retained pairwise score contributions remain spatially dependent. The
fixed-budget design controls the number of evaluated pairs exactly and induces
negative dependence within target-specific NN lists, but it does not remove
spatial dependence among nearby retained score contributions. For this reason,
we estimate the Godambe variance using a parametric score bootstrap.

Let \(\widehat{\bm\theta}\) be the fitted parameter vector and let
\(\widehat{\bm H}\) be the sensitivity matrix evaluated at
\(\widehat{\bm\theta}\). For \(b=1,\ldots,B\), a bootstrap dataset
\(Z^{*(b)}\) is simulated from the fitted model
\(P_{\widehat{\bm\theta}}\) at the same observation locations. The
deterministic NN candidate graph is reused, because it depends only on the
locations. The same thinning design and retained-pair budget as in the original
fit are then applied to each bootstrap dataset.

For each bootstrap dataset, the composite score is evaluated at the original
estimate \(\widehat{\bm\theta}\), without refitting the model:
\[
\widetilde{\bm U}^{*(b)}_{n}(\widehat{\bm\theta})
=
\sum_{(i,j)\in\mathcal E_n(m)}
\widetilde w^{*(b)}_{ij}
\bm{s}_{ij}
(\widehat{\bm\theta};Z^{*(b)}).
\]
The variability matrix is estimated by the empirical covariance of the
bootstrap scores,
\[
\widehat{\bm J}
=
\operatorname{Var}_{b}
\left\{
\widetilde{\bm U}^{*(b)}_{n}(\widehat{\bm\theta})
\right\}.
\]
Here \(\widehat{\bm H}\) and \(\widehat{\bm J}\) denote unnormalized sensitivity
and score-variability matrices, based on sums over the retained pairs. With
normalized versions, the corresponding \(K_n^{-1}\) factor must be included
explicitly. The estimated sandwich covariance matrix is
\[
\widehat{\operatorname{Var}}
(\widehat{\bm\theta})
=
\widehat{\bm H}^{-1}
\widehat{\bm J}
\widehat{\bm H}^{-1},
\]
and standard errors are computed as
\[
\widehat{\operatorname{se}}(\widehat\theta_j)
=
\left[
\widehat{\operatorname{Var}}
(\widehat{\bm\theta})_{jj}
\right]^{1/2}.
\]

This score-bootstrap procedure differs from a classical parametric bootstrap
based on refitting the model to every simulated dataset. It avoids repeated
optimizations and directly estimates the variability matrix of the composite
score, making simulation-based Godambe estimation feasible for massive spatial
datasets.

The randomized thinning is regenerated independently for each bootstrap
dataset. The resulting standard errors therefore include both the variability
of the random field and the additional Monte Carlo variability introduced by
thinning. A conditional-design version is also possible: in that case the
realized thinned edge set from the original fit is kept fixed across bootstrap
datasets, and the estimated variance is conditional on the selected edge set.
The unconditional version is used here because it matches the interpretation of
the stochastic estimators as randomized procedures.

\section{Simulation studies}
\label{sec:simulation}

Throughout the simulation studies we consider a Gaussian random field
\[
Z(\bm s)=\mu+Y(\bm s), \qquad \bm s\in[0,1]^2,
\]
where \(Y\) is a zero-mean stationary isotropic Gaussian random field with
Mat\'ern covariance function
\[
C_{\theta}(h)
=
\sigma^2 M_{\nu}(h/\alpha)
=
\sigma^2
\frac{2^{1-\nu}}{\Gamma(\nu)}
\left(\frac{h}{\alpha}\right)^\nu
K_\nu\left(\frac{h}{\alpha}\right),
\qquad h\ge 0.
\]
Here \(\mu\) is the mean, \(\alpha\) is the scale parameter,
\(\sigma^2\) is the sill, \(\nu\) is the Mat\'ern smoothness parameter,
and \(K_\nu\) denotes the modified Bessel function of the second kind. No
nugget effect is included. The parameter vector estimated in all simulation
experiments is
\[
\theta=(\mu,\alpha,\sigma^2,\nu).
\]
In all scenarios the true mean is \(\mu=0\), the true sill is \(\sigma^2=1\),
and the scale parameter \(\alpha\) is chosen to obtain the specified practical
range.

The simulation section has three complementary goals. Study~1 is an internal
calibration study within the NN pairwise likelihood family; it provides
numerical evidence on how much the deterministic NN pairwise objective can be
thinned while preserving most of the statistical efficiency of the FullNN
estimator. Study~2 uses the budget suggested by this calibration study to
compare the resulting stochastic NN pairwise estimators with a Vecchia-type
Gaussian likelihood approximation. Study~3 then isolates the additional Monte
Carlo variability introduced by the random thinning step by repeatedly thinning
the same simulated datasets. Thus, Study~1 calibrates the retained-pair budget,
Study~2 compares the calibrated estimators with an external likelihood-based
benchmark, and Study~3 quantifies the internal variability of the proposed
stochastic designs.

All computations in the simulation studies were performed on an Apple M3 Pro
machine with 18 GB of unified memory. Reported computing times should therefore
be interpreted as implementation- and hardware-dependent elapsed times.

\subsection{Simulation study 1: retained-pair budget calibration}
\label{sec:sim_budget}

The first simulation study investigates how aggressively the deterministic NN
pairwise likelihood can be thinned while preserving most of the statistical
efficiency of the FullNN estimator. The Vecchia approximation is not included
in this calibration study; it is used later as an external benchmark in
Section~\ref{sec:sim_vecchia}.

We simulated Gaussian random fields from the Mat\'ern model described above.
For each sample size, the sampling window was kept fixed, but the practical
range was reduced as \(n\) increased:
\[
(n,r_{\mathrm{prac}})\in
\{(100000,0.15),(250000,0.10),(500000,0.05)\}.
\]
Thus, although the simulations use a common spatial window, the effective range
of dependence decreases with \(n\). The resulting sequence of designs is
compatible with the weak-dependence regime underlying increasing-domain
asymptotics, because spatial dependence becomes increasingly local relative to
the observation window.

Two smoothness values were considered, \(\nu=0.5\) and \(\nu=1.5\). Since
smoothness estimation benefits from a richer set of local distances
\citep{NNP}, the NN candidate size was chosen as \(m=10\) for \(\nu=0.5\) and
\(m=50\) for \(\nu=1.5\). Table~\ref{tab:sim1_design} summarizes the six
scenarios.

For each scenario, we generated 250 independent datasets at a fixed set of
uniformly distributed sampling locations. The deterministic NN pairwise
estimator using all \(nm\) candidate pairs is denoted by FullNN. The stochastic
estimators are the Bernoulli and fixed-budget designs of
Section~\ref{sec:two_designs}. For a given candidate graph size \(m\), the
thinning fraction \(p\) is the tuning parameter of the proposed stochastic
methods. For reporting and calibration, we express the resulting budget through
\[
K_{\mathrm{tar}}/n \in \{0.10,0.25,0.50,1,2,4\},
\]
which corresponds approximately to
\[
p \simeq \frac{K_{\mathrm{tar}}/n}{m}
\]
in a regular directed \(m\)-NN graph. Thus, \(K_{\mathrm{tar}}/n=2\) corresponds
to \(p=0.20\) when \(m=10\), and to \(p=0.04\) when \(m=50\). All methods were
fitted by maximizing the pairwise marginal composite likelihood with the same
starting values and parameter bounds. For each estimator we recorded bias,
standard deviation, RMSE, computing time, speedup relative to FullNN, and
relative RMSE with respect to FullNN.

\begin{table}[t!]
\centering
\caption{Simulation Study~1 design. The scale parameter \(\alpha\) is chosen so
that the Mat\'ern practical range is \(r_{\mathrm{prac}}\). The FullNN column
reports the number of directed NN candidate pairs, equal to \(nm\).}
\label{tab:sim1_design}
\begin{tabular}{rrrrrr}
\toprule
\(n\) & \(r_{\mathrm{prac}}\) & \(\nu\) & \(\alpha\) & \(m\) & FullNN pairs \\
\midrule
100000 & 0.15 & 0.5 & 0.050071 & 10 & 1,000,000 \\
100000 & 0.15 & 1.5 & 0.031620 & 50 & 5,000,000 \\
250000 & 0.10 & 0.5 & 0.033381 & 10 & 2,500,000 \\
250000 & 0.10 & 1.5 & 0.021080 & 50 & 12,500,000 \\
500000 & 0.05 & 0.5 & 0.016690 & 10 & 5,000,000 \\
500000 & 0.05 & 1.5 & 0.010540 & 50 & 25,000,000 \\
\bottomrule
\end{tabular}
\end{table}

Table~\ref{tab:sim1_range_allparams} summarizes the effect of the retained-pair
budget on all model parameters. For each value of \(K_{\mathrm{tar}}/n\) and
for each stochastic design, the table reports the range, over the six
simulation scenarios, of the relative RMSEs with respect to FullNN. Values
close to one indicate that the stochastic estimator retains nearly the same
statistical efficiency as the deterministic NN pairwise estimator.

The results show that the mean, scale, and sill parameters are only mildly
affected by stochastic thinning, especially for moderate or large budgets. The
smoothness parameter is the most sensitive parameter. Very aggressive budgets,
such as \(K_{\mathrm{tar}}/n=0.10\) and \(0.25\), produce large losses for
smoothness estimation. This is particularly visible for Bernoulli thinning at
the smallest budget, where the retained size is random and some fits become
unstable. The degradation becomes moderate around \(K_{\mathrm{tar}}/n=1\),
while \(K_{\mathrm{tar}}/n=2\) gives a stable compromise across parameters and
scenarios. The more conservative budget \(K_{\mathrm{tar}}/n=4\) yields results
close to FullNN for all parameters, at the cost of smaller speedups.

\begin{table}[t!]
\centering
\caption{Simulation Study~1. Range, over the six simulation scenarios, of the
relative RMSEs with respect to the deterministic FullNN estimator. Values close
to one indicate nearly the same RMSE as FullNN. The table is indexed by the
retained-pair budget \(K_{\mathrm{tar}}/n\); for a given candidate graph size
\(m\), the corresponding tuning parameter \(p\) depends on \(m\).}
\label{tab:sim1_range_allparams}
\small
\begin{tabular}{rlcccc}
\toprule
\(K_{\mathrm{tar}}/n\) & Method &
\(r_\mu\) & \(r_\alpha\) & \(r_{\sigma^2}\) & \(r_\nu\) \\
\midrule
0.10 & Bernoulli    & 0.99--1.01 & 1.26--1.86 & 0.99--1.03 & 3.37--78.07 \\
0.10 & Fixed Budget & 0.99--1.01 & 1.28--1.79 & 1.00--1.02 & 3.15--4.96 \\
0.25 & Bernoulli    & 1.00--1.00 & 1.08--1.37 & 0.99--1.01 & 2.24--3.18 \\
0.25 & Fixed Budget & 1.00--1.01 & 1.08--1.38 & 0.99--1.02 & 2.30--3.12 \\
0.50 & Bernoulli    & 1.00--1.00 & 1.10--1.23 & 1.00--1.01 & 1.72--2.30 \\
0.50 & Fixed Budget & 1.00--1.00 & 1.03--1.22 & 1.00--1.00 & 1.62--2.15 \\
1.00 & Bernoulli    & 1.00--1.00 & 1.03--1.15 & 1.00--1.01 & 1.45--1.77 \\
1.00 & Fixed Budget & 1.00--1.00 & 1.02--1.13 & 1.00--1.00 & 1.41--1.72 \\
2.00 & Bernoulli    & 1.00--1.00 & 1.00--1.05 & 0.99--1.00 & 1.17--1.38 \\
2.00 & Fixed Budget & 1.00--1.00 & 1.01--1.04 & 1.00--1.00 & 1.14--1.37 \\
4.00 & Bernoulli    & 1.00--1.00 & 1.01--1.03 & 1.00--1.00 & 1.09--1.18 \\
4.00 & Fixed Budget & 1.00--1.00 & 0.99--1.02 & 1.00--1.00 & 1.05--1.19 \\
\bottomrule
\end{tabular}
\end{table}

Table~\ref{tab:sim1_k2_allparams} gives detailed results for
\(K_{\mathrm{tar}}/n=2\). The stochastic estimators evaluate only \(2n\)
bivariate likelihood terms, independently of \(m\). For the rough case
\(\nu=0.5\), this corresponds to 20\% of the \(m=10\) candidate graph. For the
smoother case \(\nu=1.5\), it corresponds to only 4\% of the \(m=50\)
candidate graph. The relative RMSEs for the mean, scale, and sill remain close
to one, whereas the largest loss of efficiency is concentrated on the
smoothness parameter.

\begin{table}[t!]
\centering
\caption{Simulation Study~1 results at \(K_{\mathrm{tar}}/n=2\). The column
\(p\) reports the tuning parameter corresponding to this budget for the given
candidate graph size \(m\). Relative RMSEs are computed with respect to FullNN
in the same scenario.}
\label{tab:sim1_k2_allparams}
\small
\begin{tabular}{rrrlrrrrrr}
\toprule
\(n\) & \(\nu\) & \(m\) & Method & \(p\) & Speedup &
\(r_\mu\) & \(r_\alpha\) & \(r_{\sigma^2}\) & \(r_\nu\) \\
\midrule
100000 & 0.5 & 10 & Bernoulli    & 0.20 &  4.64 & 1.00 & 1.01 & 1.00 & 1.28 \\
100000 & 0.5 & 10 & Fixed Budget & 0.20 &  4.74 & 1.00 & 1.02 & 1.00 & 1.37 \\
\addlinespace
100000 & 1.5 & 50 & Bernoulli    & 0.04 & 26.71 & 1.00 & 1.05 & 1.00 & 1.29 \\
100000 & 1.5 & 50 & Fixed Budget & 0.04 & 26.72 & 1.00 & 1.01 & 1.00 & 1.14 \\
\addlinespace
250000 & 0.5 & 10 & Bernoulli    & 0.20 &  5.11 & 1.00 & 1.00 & 1.00 & 1.38 \\
250000 & 0.5 & 10 & Fixed Budget & 0.20 &  5.34 & 1.00 & 1.01 & 1.00 & 1.32 \\
\addlinespace
250000 & 1.5 & 50 & Bernoulli    & 0.04 & 27.81 & 1.00 & 1.01 & 1.00 & 1.17 \\
250000 & 1.5 & 50 & Fixed Budget & 0.04 & 25.82 & 1.00 & 1.04 & 1.00 & 1.27 \\
\addlinespace
500000 & 0.5 & 10 & Bernoulli    & 0.20 &  5.39 & 1.00 & 1.03 & 0.99 & 1.36 \\
500000 & 0.5 & 10 & Fixed Budget & 0.20 &  4.54 & 1.00 & 1.02 & 1.00 & 1.37 \\
\addlinespace
500000 & 1.5 & 50 & Bernoulli    & 0.04 & 31.55 & 1.00 & 1.04 & 1.00 & 1.18 \\
500000 & 1.5 & 50 & Fixed Budget & 0.04 & 22.57 & 1.00 & 1.02 & 1.00 & 1.18 \\
\bottomrule
\end{tabular}
\end{table}

The comparison between Bernoulli and fixed-budget thinning shows that the two
designs have similar overall statistical behavior at the same retained-pair
budget. Bernoulli thinning is simpler but controls the retained size only in
expectation, whereas fixed-budget thinning enforces the exact retained-pair
budget and distributes retained pairs more regularly across target-specific NN
lists. This additional control does not imply uniform RMSE dominance. For the
mean, scale, and sill parameters, the two stochastic estimators are typically
very close; for smoothness, fixed-budget thinning improves performance in some
scenarios but not uniformly. Thus, the main role of fixed-budget thinning is
exact computational control and a more regular use of the NN candidate graph,
rather than guaranteed efficiency improvement.

Overall, Study~1 provides numerical support for using
\(K_{\mathrm{tar}}/n=2\) as a practical budget for the Mat\'ern simulation
settings considered here when smoothness is estimated. This value should be
interpreted as an empirical calibration for the present class of designs, rather
than as a universal recommendation.

For the proposed stochastic methods, the corresponding values of \(p\) are
reported in Table~\ref{tab:sim1_k2_allparams}. This budget gives large
computational savings, especially when \(m=50\). At this budget, the mean,
scale, and sill parameters remain very close to the deterministic FullNN
benchmark, while the loss of efficiency is concentrated mainly on the
smoothness parameter. The budget \(K_{\mathrm{tar}}/n=4\) is a more
conservative option when smoothness estimation is the primary objective.
Budgets below \(K_{\mathrm{tar}}/n=0.5\) are useful as stress tests, but are too
aggressive for reliable smoothness estimation in the present experiments.

\subsection{Simulation study 2: comparison with Vecchia approximation}
\label{sec:sim_vecchia}

The second simulation study compares the proposed stochastic NN pairwise
estimators with a Vecchia-type Gaussian likelihood approximation. We use the
same six scenarios as in Study~1. Based on the budget-calibration results, the
stochastic NN pairwise methods are run at the retained-pair budget
\[
K_{\mathrm{tar}}/n = 2.
\]
For a given candidate graph size \(m\), this budget corresponds to a specific
value of the thinning parameter \(p\). The candidate NN size is again \(m=10\)
for \(\nu=0.5\) and \(m=50\) for \(\nu=1.5\), giving \(p=0.20\) and
\(p=0.04\), respectively. In the fixed-budget design this corresponds, in the
regular directed NN graph, to two retained pairs per target location.

The stochastic NN pairwise estimators were fitted with \texttt{GeoFit} from
\texttt{GeoModels}, which directly maximizes the selected weighted pairwise
composite likelihood. The Vecchia benchmark was fitted with 30 conditioning
neighbors using \texttt{fit\_model} from the \texttt{GpGp} package for
\textsf{R} \citep{GpGp}. The two implementations differ in several respects:
\texttt{fit\_model} uses a grouped Vecchia implementation, profiles the linear
mean parameters during covariance-parameter optimization, and applies Fisher
scoring, whereas \texttt{GeoFit} directly optimizes the selected pairwise
composite likelihood without profiling the mean parameters. The reported
timings should therefore be interpreted as practical implementation-level
timings for the specific workflows used here, rather than as an optimizer- or
profiling-matched comparison of the underlying estimating criteria.

It is also important to note that both computational workflows could be further
accelerated by using more specialized implementations. For example, the
evaluation of large collections of local likelihood or covariance contributions
is naturally amenable to parallelization, and GPU-based strategies can provide
substantial speedups for large-scale spatial likelihood computations
\citep{James2024}. Therefore, the timings reported below should be
viewed as representative of the serial or standard implementations used in this
study, rather than as hardware-optimized lower bounds for either approach.

For each parameter, we report RMSE ratios with respect to Vecchia,
\[
r_j(\mathrm{method})
=
\frac{\mathrm{RMSE}_j(\mathrm{method})}
     {\mathrm{RMSE}_j(\mathrm{Vecchia})},
\qquad
j\in\{\mu,\alpha,\sigma^2,\nu\}.
\]
Values larger than one indicate larger RMSE than Vecchia. We also report the
computing-time speedup with respect to Vecchia,
\[
S_{\mathrm{Vec}}(\mathrm{method})
=
\frac{T_{\mathrm{Vecchia}}}{T_{\mathrm{method}}},
\]
which should be interpreted as an implementation-level speedup for the specific
workflows described above.

\begin{table}[t!]
\centering
\caption{Simulation Study~2: relative performance with respect to the Vecchia
approximation. The column \(p\) reports the thinning parameter for the
stochastic NN pairwise estimators; it is not applicable to FullNN. The columns
\(r_\mu\), \(r_\alpha\), \(r_{\sigma^2}\), and \(r_\nu\) report RMSE ratios
relative to Vecchia. Values larger than one indicate larger RMSE than Vecchia.
The column \(S_{\mathrm{Vec}}\) reports the implementation-level computing-time
speedup relative to the \texttt{GpGp} Vecchia fit. Vecchia uses 30 conditioning
neighbors.}
\label{tab:study2_vecchia_relative}
\small
\begin{tabular}{rrrlcrrrrr}
\toprule
\(n\) & \(\nu\) & \(m\) & Method & \(p\) &
\(S_{\mathrm{Vec}}\) &
\(r_\mu\) & \(r_\alpha\) & \(r_{\sigma^2}\) & \(r_\nu\) \\
\midrule
100000 & 0.5 & 10 & FullNN       & --   &  2.23 & 1.03 & 1.12 & 1.09 &  1.28 \\
100000 & 0.5 & 10 & Bernoulli    & 0.20 & 10.73 & 1.03 & 1.13 & 1.09 &  1.81 \\
100000 & 0.5 & 10 & Fixed Budget & 0.20 & 10.80 & 1.03 & 1.13 & 1.09 &  1.81 \\
\addlinespace
100000 & 1.5 & 50 & FullNN       & --   &  0.75 & 1.05 & 2.37 & 1.27 &  8.80 \\
100000 & 1.5 & 50 & Bernoulli    & 0.04 & 18.16 & 1.05 & 2.47 & 1.27 & 10.96 \\
100000 & 1.5 & 50 & Fixed Budget & 0.04 & 18.33 & 1.05 & 2.49 & 1.27 & 10.92 \\
\addlinespace
250000 & 0.5 & 10 & FullNN       & --   &  2.23 & 1.07 & 1.10 & 1.06 &  1.37 \\
250000 & 0.5 & 10 & Bernoulli    & 0.20 & 11.38 & 1.07 & 1.13 & 1.06 &  2.01 \\
250000 & 0.5 & 10 & Fixed Budget & 0.20 & 11.26 & 1.07 & 1.13 & 1.06 &  1.87 \\
\addlinespace
250000 & 1.5 & 50 & FullNN       & --   &  1.68 & 1.06 & 1.81 & 0.97 &  8.13 \\
250000 & 1.5 & 50 & Bernoulli    & 0.04 & 37.51 & 1.06 & 1.94 & 0.98 & 10.71 \\
250000 & 1.5 & 50 & Fixed Budget & 0.04 & 37.51 & 1.06 & 1.90 & 0.97 & 10.16 \\
\addlinespace
500000 & 0.5 & 10 & FullNN       & --   &  2.82 & 1.02 & 1.03 & 1.01 &  1.25 \\
500000 & 0.5 & 10 & Bernoulli    & 0.20 & 13.75 & 1.02 & 1.07 & 1.02 &  1.73 \\
500000 & 0.5 & 10 & Fixed Budget & 0.20 & 13.96 & 1.02 & 1.05 & 1.01 &  1.73 \\
\addlinespace
500000 & 1.5 & 50 & FullNN       & --   &  1.42 & 1.02 & 1.95 & 0.93 &  5.23 \\
500000 & 1.5 & 50 & Bernoulli    & 0.04 & 32.49 & 1.02 & 2.02 & 0.93 &  6.34 \\
500000 & 1.5 & 50 & Fixed Budget & 0.04 & 32.94 & 1.02 & 2.08 & 0.93 &  6.83 \\
\bottomrule
\end{tabular}
\end{table}

Table~\ref{tab:study2_vecchia_relative} shows a clear
statistical--computational trade-off. Vecchia is generally the most accurate
method for estimating the smoothness parameter, especially when \(\nu=1.5\).
For the mean and sill parameters, the differences between Vecchia and the
pairwise methods are much smaller. For the scale parameter, Vecchia is also
more accurate in the smoother scenarios, although the discrepancy is less
extreme than for the smoothness parameter.

\begin{figure}[t!]
\centering
\includegraphics[width=0.95\textwidth]{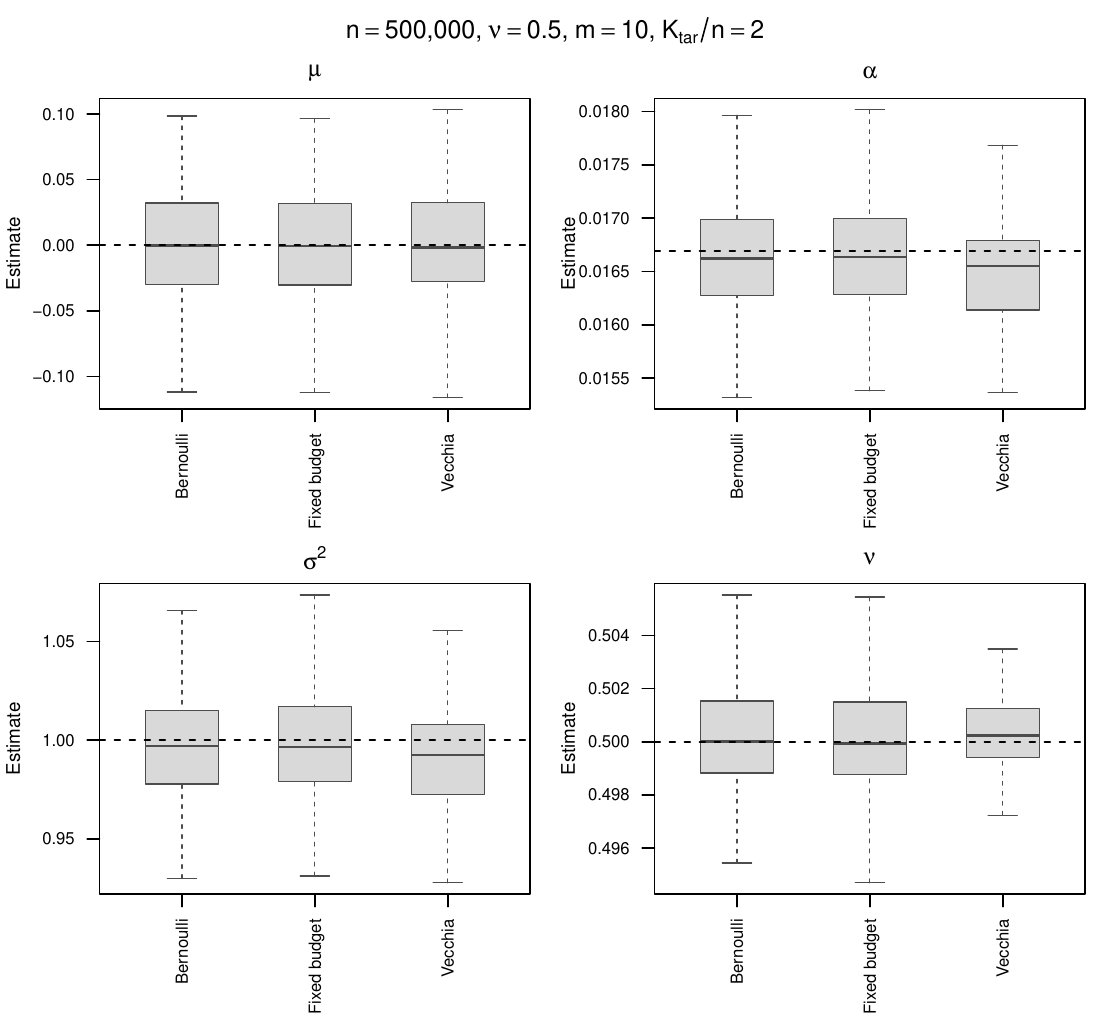}
\caption{Simulation Study~2. Boxplots of parameter estimates for the largest
sample size in the rough Mat\'ern scenario: \(n=500{,}000\), \(\nu=0.5\),
\(m=10\), and \(K_{\mathrm{tar}}/n=2\). Dashed horizontal lines indicate the
true parameter values.}
\label{fig:sim2_boxplots_nu05}
\end{figure}

\begin{figure}[t!]
\centering
\includegraphics[width=0.95\textwidth]{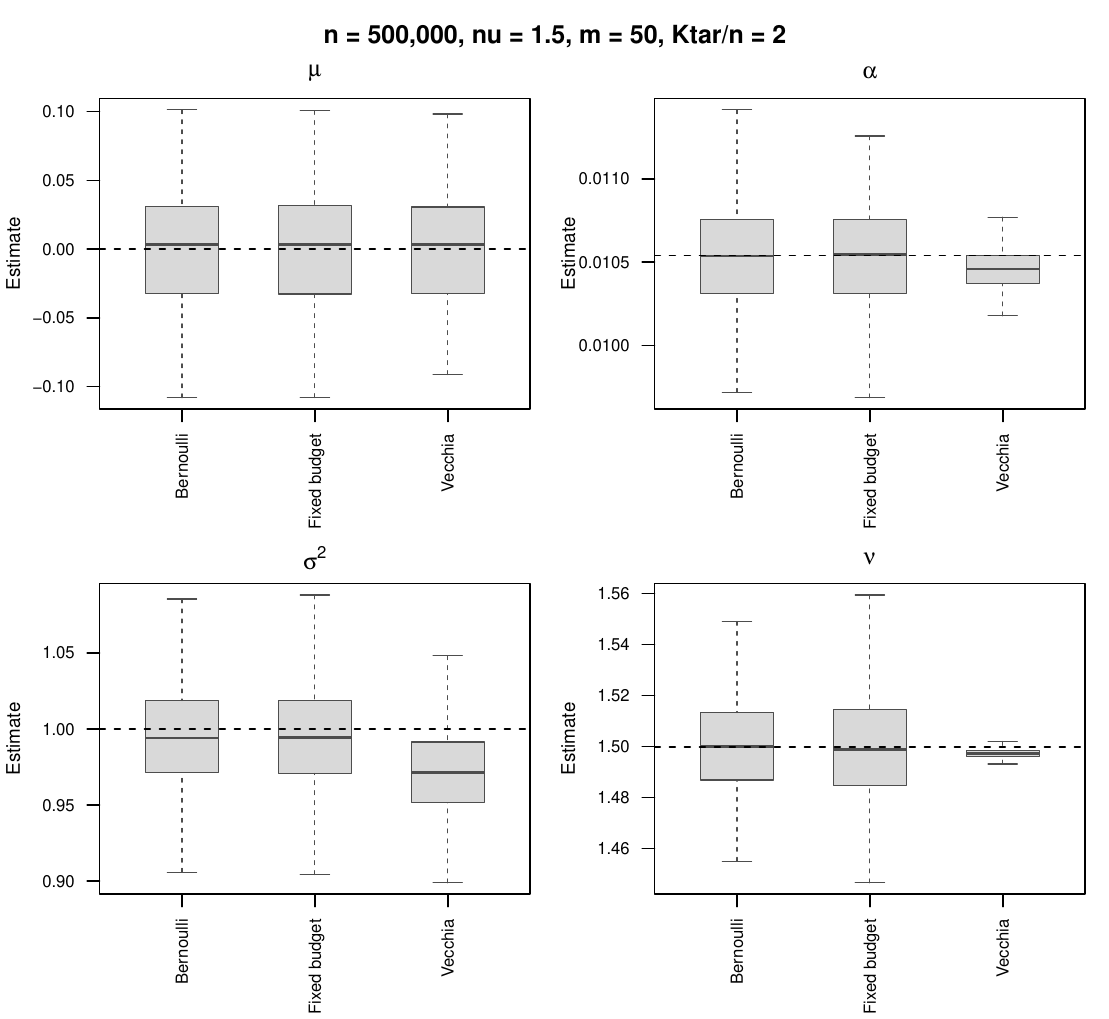}
\caption{Simulation Study~2. Boxplots of parameter estimates for the largest
sample size in the smoother Mat\'ern scenario: \(n=500{,}000\), \(\nu=1.5\),
\(m=50\), and \(K_{\mathrm{tar}}/n=2\). Dashed horizontal lines indicate the
true parameter values.}
\label{fig:sim2_boxplots_nu15}
\end{figure}

Figures~\ref{fig:sim2_boxplots_nu05} and~\ref{fig:sim2_boxplots_nu15} provide
a graphical comparison for the largest sample size. In the rough scenario,
\(\nu=0.5\), the stochastic NN pairwise estimators and Vecchia have very
similar distributions for all parameters. In the smoother scenario,
\(\nu=1.5\), the distributions remain similar for the mean and sill parameters,
whereas Vecchia is more concentrated for the scale and, especially, for the
smoothness parameter. This agrees with the RMSE ratios in
Table~\ref{tab:study2_vecchia_relative}: the main loss of efficiency of the
stochastic NN pairwise estimators relative to Vecchia is concentrated on
smoothness estimation.

In the implementation-level comparison considered here, the computational
differences are substantial. For \(\nu=0.5\), the stochastic pairwise fits
required about 11--14 times less covariance-fitting time than the Vecchia
benchmark. For \(\nu=1.5\), where the deterministic candidate graph uses
\(m=50\), the difference is even larger: about 18-fold for \(n=100000\), about
37--38-fold for \(n=250000\), and about 32--33-fold for \(n=500000\). These
results show that the stochastic NN pairwise estimators can be attractive
practical alternatives to the Vecchia benchmark in large spatial datasets when
covariance-fitting time is a primary constraint. The price paid for the large
speedups is not uniform across parameters: it is small for the mean and sill,
moderate for the scale in some smoother scenarios, and largest for the
Mat\'ern smoothness parameter.

The two proposed stochastic designs behave very similarly. Bernoulli thinning
has random retained size but is simple to implement, whereas fixed-budget
thinning enforces the retained-pair budget exactly and allocates it across the
target-specific NN lists. Across the six scenarios, neither stochastic design
uniformly dominates the other in terms of RMSE. At \(K_{\mathrm{tar}}/n=2\),
the fixed-budget design is slightly better for smoothness in some scenarios
and slightly worse or essentially identical in others. Its main practical
advantage is therefore exact budget control and a more regular allocation of
the retained pairs, rather than systematic RMSE dominance.

Overall, Study~2 shows that, in the practical implementation-level comparison
considered here, the proposed stochastic NN pairwise estimators provide
substantial reductions in covariance-fitting time relative to the Vecchia
benchmark when some loss of efficiency is acceptable. The loss is small for the
mean and sill, moderate for the scale in some smoother scenarios, and largest
for the Mat\'ern smoothness parameter.

\FloatBarrier

\subsection{Simulation study 3: internal variability of stochastic thinning}
\label{sec:sim_internal}

The third simulation study investigates the additional variability introduced by
the random thinning step itself. Unlike Studies~1 and~2, where each simulated
dataset was associated with a single thinning realization, here we condition on
the same simulated dataset and repeatedly generate independent thinnings. This
separates the variability due to the Gaussian random field from the variability
due only to the stochastic pair-selection mechanism.

We use the same six scenarios as in the previous studies and fix the retained
budget at
\[
K_{\mathrm{tar}}/n=2.
\]
As in Study~2, this budget corresponds to \(p=0.20\) when \(m=10\) and to
\(p=0.04\) when \(m=50\). For each scenario, we generate
\(N_{\mathrm{sim}}=50\) independent datasets. For each dataset, FullNN is
fitted once, while Bernoulli thinning and fixed-budget thinning are each
repeated 30 times independently. Vecchia is not included, because the goal is
to compare the internal stochastic variability of the two proposed thinning
designs.

Let \(\widehat\theta^{(k,b)}_{M,j}\) be the estimate of parameter \(j\) for
dataset \(k\), thinning replicate \(b\), and method
\(M\in\{\mathrm{Bernoulli},\mathrm{Fixed\ Budget}\}\). For each dataset, we
compute the conditional variance over thinning replicates,
\[
s^2_{M,j}(k)
=
\frac{1}{B-1}
\sum_{b=1}^{B}
\left(
\widehat\theta^{(k,b)}_{M,j}
-
\bar\theta^{(k)}_{M,j}
\right)^2,
\qquad
\bar\theta^{(k)}_{M,j}
=
\frac{1}{B}
\sum_{b=1}^{B}
\widehat\theta^{(k,b)}_{M,j}.
\]
We summarize the fraction of total stochastic-estimator variance attributable
to thinning as
\[
\rho_{M,j}
=
\frac{
N_{\mathrm{sim}}^{-1}\sum_{k=1}^{N_{\mathrm{sim}}} s^2_{M,j}(k)
}{
\mathrm{Var}_{k,b}\left(\widehat\theta^{(k,b)}_{M,j}\right)
}.
\]
Values of \(\rho_{M,j}\) close to zero indicate that the random thinning step
contributes little relative to the overall variability of the estimator.

Table~\ref{tab:study3_internal} reports \(100\rho_{M,j}\). The internal
thinning variability is negligible for the mean parameter. For Bernoulli
thinning it is about \(0.02\%\)--\(0.06\%\), whereas for fixed-budget thinning
it is essentially zero at the scale of the table. The sill parameter shows the
same qualitative behavior: the internal thinning component is very small for
both methods, but is systematically smaller for the fixed-budget design.

The scale parameter is more affected by the thinning mechanism. Across the six
scenarios, the internal thinning component accounts for about \(3\%\)--\(9\%\)
of the total stochastic-estimator variance. The two thinning schemes are very
similar for this parameter, although fixed-budget thinning is slightly smaller
in most scenarios. The smoothness parameter is the most sensitive parameter:
the internal thinning component accounts for about \(28\%\)--\(51\%\) of the
total stochastic-estimator variance. Thus, the additional Monte Carlo
variability induced by thinning is concentrated mainly on smoothness
estimation.

\begin{table}[t!]
\centering
\caption{
Simulation Study~3: internal variability of stochastic thinning at
\(K_{\mathrm{tar}}/n=2\). Entries are \(100\rho_{M,j}\), where
\(\rho_{M,j}\) is the fraction of total stochastic-estimator variance
attributable to the random thinning step. Results are based on 50 datasets and
30 independent thinnings per dataset for each scenario.
}
\label{tab:study3_internal}
\small
\begin{tabular}{rrlrrrr}
\toprule
\(n\) & \(\nu\) & Method &
\(100\rho_\mu\) &
\(100\rho_\alpha\) &
\(100\rho_{\sigma^2}\) &
\(100\rho_\nu\) \\
\midrule
100000 & 0.5 & Bernoulli    & 0.026 & 2.96 & 0.081 & 38.48 \\
100000 & 0.5 & Fixed Budget & 0.001 & 2.92 & 0.005 & 40.04 \\
\addlinespace
100000 & 1.5 & Bernoulli    & 0.028 & 7.75 & 0.095 & 28.38 \\
100000 & 1.5 & Fixed Budget & 0.000 & 7.78 & 0.003 & 29.40 \\
\addlinespace
250000 & 0.5 & Bernoulli    & 0.026 & 3.22 & 0.097 & 50.87 \\
250000 & 0.5 & Fixed Budget & 0.001 & 3.15 & 0.005 & 49.75 \\
\addlinespace
250000 & 1.5 & Bernoulli    & 0.024 & 9.03 & 0.069 & 35.32 \\
250000 & 1.5 & Fixed Budget & 0.000 & 8.19 & 0.002 & 33.37 \\
\addlinespace
500000 & 0.5 & Bernoulli    & 0.027 & 4.47 & 0.136 & 49.85 \\
500000 & 0.5 & Fixed Budget & 0.001 & 4.39 & 0.009 & 50.11 \\
\addlinespace
500000 & 1.5 & Bernoulli    & 0.056 & 9.23 & 0.122 & 40.37 \\
500000 & 1.5 & Fixed Budget & 0.001 & 8.77 & 0.006 & 39.22 \\
\bottomrule
\end{tabular}
\end{table}

Table~\ref{tab:study3_rmse} reports the corresponding relative RMSEs with
respect to FullNN. At \(K_{\mathrm{tar}}/n=2\), both stochastic estimators are
essentially indistinguishable from FullNN for the mean and sill parameters.
The scale parameter shows only mild inflation, while the main efficiency loss
again concerns the smoothness parameter. Across the six scenarios, the relative
RMSE for smoothness ranges from about 1.18 to 1.38.

\begin{table}[t!]
\centering
\caption{
Simulation Study~3: relative RMSEs of the stochastic estimators with respect to
FullNN at \(K_{\mathrm{tar}}/n=2\). The speedup is computed with respect to
FullNN in the same scenario.
}
\label{tab:study3_rmse}
\small
\begin{tabular}{rrlrrrrr}
\toprule
\(n\) & \(\nu\) & Method & Speedup &
\(r_\mu\) & \(r_\alpha\) & \(r_{\sigma^2}\) & \(r_\nu\) \\
\midrule
100000 & 0.5 & Bernoulli    &  4.84 & 1.00 & 1.02 & 1.00 & 1.31 \\
100000 & 0.5 & Fixed Budget &  4.83 & 1.00 & 1.02 & 1.00 & 1.29 \\
\addlinespace
100000 & 1.5 & Bernoulli    & 22.64 & 1.00 & 1.05 & 1.00 & 1.21 \\
100000 & 1.5 & Fixed Budget & 22.65 & 1.00 & 1.04 & 1.00 & 1.20 \\
\addlinespace
250000 & 0.5 & Bernoulli    &  4.82 & 1.00 & 1.02 & 1.00 & 1.35 \\
250000 & 0.5 & Fixed Budget &  4.82 & 1.00 & 1.01 & 1.00 & 1.37 \\
\addlinespace
250000 & 1.5 & Bernoulli    & 22.44 & 1.00 & 1.05 & 1.00 & 1.21 \\
250000 & 1.5 & Fixed Budget & 22.37 & 1.00 & 1.03 & 1.00 & 1.18 \\
\addlinespace
500000 & 0.5 & Bernoulli    &  4.74 & 1.00 & 1.03 & 1.00 & 1.38 \\
500000 & 0.5 & Fixed Budget &  4.82 & 1.00 & 1.02 & 1.00 & 1.38 \\
\addlinespace
500000 & 1.5 & Bernoulli    & 23.19 & 1.00 & 1.05 & 1.00 & 1.34 \\
500000 & 1.5 & Fixed Budget & 23.44 & 1.00 & 1.06 & 1.00 & 1.34 \\
\bottomrule
\end{tabular}
\end{table}

The comparison between Bernoulli and fixed-budget thinning shows two main
patterns. For the mean and sill parameters, fixed-budget thinning reduces the
internal thinning component relative to Bernoulli thinning. This is consistent
with the exact retained-pair budget and the more regular allocation across
target-specific NN lists. For the scale and smoothness parameters, however, the
two thinning schemes behave similarly. In particular, exact budget control does
not remove the main source of additional variability in smoothness estimation,
which appears to be related to which local distances are selected rather than
only to how many pairs are retained.

Overall, Study~3 reinforces the conclusions of the previous studies. At the
calibrated budget \(K_{\mathrm{tar}}/n=2\), the additional randomness due to
thinning is negligible for the mean and sill, moderate for the scale parameter,
and most relevant for smoothness. Fixed-budget thinning provides exact control
of the retained-pair count and reduces the internal thinning component for the
mean and sill, but it should not be interpreted as uniformly reducing RMSE or
internal variability for all parameters.

\section{WorldClim temperature application}
\label{sec:application}

We now consider a large gridded temperature application. The goal of this
example is to evaluate the proposed stochastic NN pairwise likelihood in a
realistic multi-million-observation setting and to compare it with a
Vecchia-type benchmark in terms of covariance-fitting time and predictive
accuracy.

The data were obtained from the WorldClim v2.1 monthly climate grids
(\href{https://www.worldclim.org/data/worldclim21.html}{WorldClim v2.1}).
We analyze July average temperature over a western--central United States
window,
$
[-125,-100]\times[30,50],
$
in longitude--latitude coordinates. From the valid grid cells in this window we
used a random sample of \(2{,}500{,}000\) observations. Temperature was analyzed
on its original Celsius scale, without logarithmic transformation.

Longitude--latitude coordinates were converted to projected coordinates in
kilometers using a local equirectangular projection centered on the selected
window. The projected coordinates were used only to define distances and
nearest-neighbor sets. A random 90\%--10\% split was then used, giving
$
n_{\mathrm{train}}=2{,}250{,}000
$
observations for model fitting and \(250{,}000\) observations for prediction
assessment.

Unlike the simulation studies, the temperature application contains a
pronounced large-scale spatial trend. We therefore fitted a smooth
two-dimensional GAM mean surface to the training data, using a basis dimension
\(k=500\). The fitted GAM mean explained about \(89\%\) of the marginal
variation and was kept fixed in the covariance-estimation step. The residual
field was modelled as a stationary isotropic Mat\'ern Gaussian random field
with zero nugget. The pairwise likelihood was fitted to the original training
responses using the GAM fitted values as a spatially varying mean. For
prediction, the GAM mean surface was evaluated at the test locations and local
kriging was applied using 150 nearest training observations.

Motivated by the budget-calibration results in Section~\ref{sec:sim_budget},
we considered retained-pair budgets
\[
K_{\mathrm{tar}}/n_{\mathrm{train}}\in\{1,2\}.
\]
The value \(K_{\mathrm{tar}}/n_{\mathrm{train}}=2\) is the empirically
calibrated budget suggested by the simulation study, whereas
\(K_{\mathrm{tar}}/n_{\mathrm{train}}=1\) is included to assess whether a more
aggressive thinning level is sufficient for prediction in this large gridded
application. We considered both thinning designs of
Section~\ref{sec:two_designs} and two candidate NN graph sizes,
\(m\in\{10,50\}\). For a regular directed \(m\)-NN graph, the corresponding
nominal retained fraction is approximately
\[
p \simeq \frac{K_{\mathrm{tar}}/n_{\mathrm{train}}}{m},
\]
and the resulting values of \(p\) are reported in
Table~\ref{tab:app_temperature}. Fixed-budget thinning retains exactly
\(K_{\mathrm{tar}}\) pairs, whereas Bernoulli thinning retains approximately
that number in expectation.

As an external benchmark, we fitted a Vecchia-type Gaussian approximation with
30 conditioning neighbors. The stochastic NN pairwise fits were computed with
\texttt{GeoFit} from \texttt{GeoModels}, which directly maximizes the selected
weighted pairwise composite likelihood without profiling the mean parameters.
The Vecchia benchmark was computed with \texttt{fit\_model} from
\texttt{GpGp}, which uses a grouped Vecchia implementation, profiles the linear
mean parameters during covariance-parameter optimization, and applies Fisher
scoring. The Vecchia approximation was fitted to the GAM residuals with an
intercept. The estimated intercept was \(\widehat\beta=-0.0104\), and this
value was added to the GAM mean surface when computing Vecchia-based
predictions. Apart from this intercept correction, the same local kriging
prediction machinery was used.

Table~\ref{tab:app_temperature} reports the number of retained pairwise
contributions, covariance-fitting time, and predictive accuracy on the test
set. The predictive scores shown are the mean absolute error (MAE), root mean
squared prediction error (RMSPE), and continuous ranked probability score
(CRPS). The fitting times exclude the common train--test preprocessing and the
common GAM mean-surface fit. They should be interpreted as practical
implementation-level timings for the specific \textsf{R} workflows used here,
rather than as an optimizer- or profiling-matched comparison of the underlying
estimating criteria.

\begin{table}[t!]
\centering
\caption{Temperature application based on WorldClim July average temperature.
Predictive performance and covariance-fitting time for the stochastic NN
pairwise estimators and the Vecchia benchmark. For the proposed methods, \(m\)
denotes the NN candidate graph size; for Vecchia, the entry 30 denotes the
number of conditioning neighbors. Fitting times are implementation-level
timings for the specific \textsf{R} workflows used here.}
\label{tab:app_temperature}
\small
\begin{tabular}{lrrrrrrrr}
\toprule
Method & \(K_{\mathrm{tar}}/n\) & \(m\) & \(p\) &
Pairs \((10^6)\) & Fit time (s) & MAE & RMSPE & CRPS \\
\midrule
Bernoulli    & 1 & 10 & 0.10 & 2.251 & 154.6 & 0.10239 & 0.18268 & 0.08807 \\
Fixed Budget & 1 & 10 & 0.10 & 2.250 & 141.0 & 0.10237 & 0.18261 & 0.08789 \\
Bernoulli    & 1 & 50 & 0.02 & 2.250 & 158.8 & 0.10247 & 0.18293 & 0.08870 \\
Fixed Budget & 1 & 50 & 0.02 & 2.250 & 126.0 & 0.10244 & 0.18285 & 0.08848 \\
\addlinespace
Bernoulli    & 2 & 10 & 0.20 & 4.500 & 299.9 & 0.10238 & 0.18263 & 0.08792 \\
Fixed Budget & 2 & 10 & 0.20 & 4.500 & 318.5 & 0.10238 & 0.18265 & 0.08796 \\
Bernoulli    & 2 & 50 & 0.04 & 4.502 & 280.3 & 0.10243 & 0.18280 & 0.08830 \\
Fixed Budget & 2 & 50 & 0.04 & 4.500 & 259.9 & 0.10244 & 0.18285 & 0.08847 \\
\addlinespace
Vecchia      & -- & 30 & --   & --    & 3088.7 & 0.10230 & 0.18232 & 0.08759 \\
\bottomrule
\end{tabular}
\end{table}

The empirical 95\% coverage was close to nominal for all methods. Across the
stochastic NN pairwise fits it ranged from 0.9408 to 0.9444, while Vecchia gave
coverage 0.9399. The corresponding 95\% interval scores were also very similar,
ranging from 1.2078 to 1.2088 for the stochastic fits and equal to 1.2069 for
Vecchia.

Table~\ref{tab:app_temperature_params} reports the fitted Mat\'ern covariance
parameters for the residual field after removal of the GAM mean surface.
For the stochastic NN pairwise fits, standard errors were computed using the
parametric score bootstrap described in Section~\ref{sec:inference}, with
\(100\) bootstrap score evaluations and the spectral turning-bands simulator
proposed in \citealp{newfast}, using \(L=1000\). Standard errors are reported
only for the stochastic NN pairwise fits; Vecchia is used as an external
predictive and computational benchmark.

\begin{table}[t!]
\centering
\caption{Temperature application based on WorldClim July average temperature.
Estimated Mat\'ern covariance parameters for the residual field after fitting
the GAM mean surface. Standard errors for the stochastic NN pairwise fits are
reported in parentheses and were computed using the parametric score bootstrap
with 100 score evaluations and the turning-bands simulator with \(L=1000\).
Standard errors are not reported for the Vecchia benchmark. The nugget
parameter is fixed at zero.}
\label{tab:app_temperature_params}
\small
\begin{tabular}{lrrrrrr}
\toprule
Method & \(K_{\mathrm{tar}}/n\) & \(m\) & \(p\) &
\(\widehat\alpha\) & \(\widehat\sigma^2\) & \(\widehat\nu\) \\
\midrule
Bernoulli    & 1 & 10 & 0.10 &
11.414 {\scriptsize (3.954)} &
2.568 {\scriptsize (0.094)} &
0.787 {\scriptsize (0.169)} \\

Fixed Budget & 1 & 10 & 0.10 &
11.219 {\scriptsize (4.098)} &
2.569 {\scriptsize (0.091)} &
0.794 {\scriptsize (0.168)} \\

Bernoulli    & 1 & 50 & 0.02 &
12.070 {\scriptsize (2.007)} &
2.571 {\scriptsize (0.091)} &
0.764 {\scriptsize (0.095)} \\

Fixed Budget & 1 & 50 & 0.02 &
11.873 {\scriptsize (2.293)} &
2.568 {\scriptsize (0.083)} &
0.771 {\scriptsize (0.114)} \\
\addlinespace

Bernoulli    & 2 & 10 & 0.20 &
11.291 {\scriptsize (3.029)} &
2.566 {\scriptsize (0.090)} &
0.792 {\scriptsize (0.124)} \\

Fixed Budget & 2 & 10 & 0.20 &
11.335 {\scriptsize (3.492)} &
2.568 {\scriptsize (0.083)} &
0.791 {\scriptsize (0.154)} \\

Bernoulli    & 2 & 50 & 0.04 &
11.799 {\scriptsize (2.008)} &
2.569 {\scriptsize (0.086)} &
0.775 {\scriptsize (0.099)} \\

Fixed Budget & 2 & 50 & 0.04 &
11.914 {\scriptsize (2.104)} &
2.568 {\scriptsize (0.099)} &
0.770 {\scriptsize (0.096)} \\
\addlinespace

Vecchia      & -- & 30 & -- &
9.091 &
2.213 &
0.831 \\
\bottomrule
\end{tabular}
\end{table}

The stochastic NN pairwise estimates are stable across thinning designs,
retained-pair budgets, and candidate graph sizes. Across the stochastic fits,
the estimated sill remains essentially unchanged, around \(2.57\), while the
Mat\'ern scale parameter ranges from about \(11.2\) to \(12.1\) km and the
smoothness parameter from about \(0.76\) to \(0.79\). Increasing
\(K_{\mathrm{tar}}/n_{\mathrm{train}}\) from 1 to 2 has little effect on the
fitted covariance parameters, and Bernoulli and fixed-budget thinning lead to
nearly indistinguishable estimates. The Vecchia fit gives a smaller scale and
sill and a slightly larger smoothness, reflecting the fact that it optimizes a
different approximate likelihood criterion and uses a different fitting
workflow.

As a diagnostic for the fitted residual covariance model,
Figure~\ref{fig:app_variogram} compares the empirical semivariogram of the GAM
residuals with the Mat\'ern semivariogram fitted by the Bernoulli stochastic NN
pairwise likelihood with \(K_{\mathrm{tar}}/n_{\mathrm{train}}=1\), \(m=10\),
and \(p=0.10\). This aggressive configuration retained approximately
\(2.251\times10^6\) pairwise contributions. The fitted curve closely follows
the empirical semivariogram over the main range of spatial lags, indicating
that the thinned pairwise fit captures the residual spatial dependence left
after removal of the large-scale GAM mean surface.

\begin{figure}[t!]
\centering
\includegraphics[width=0.72\textwidth]{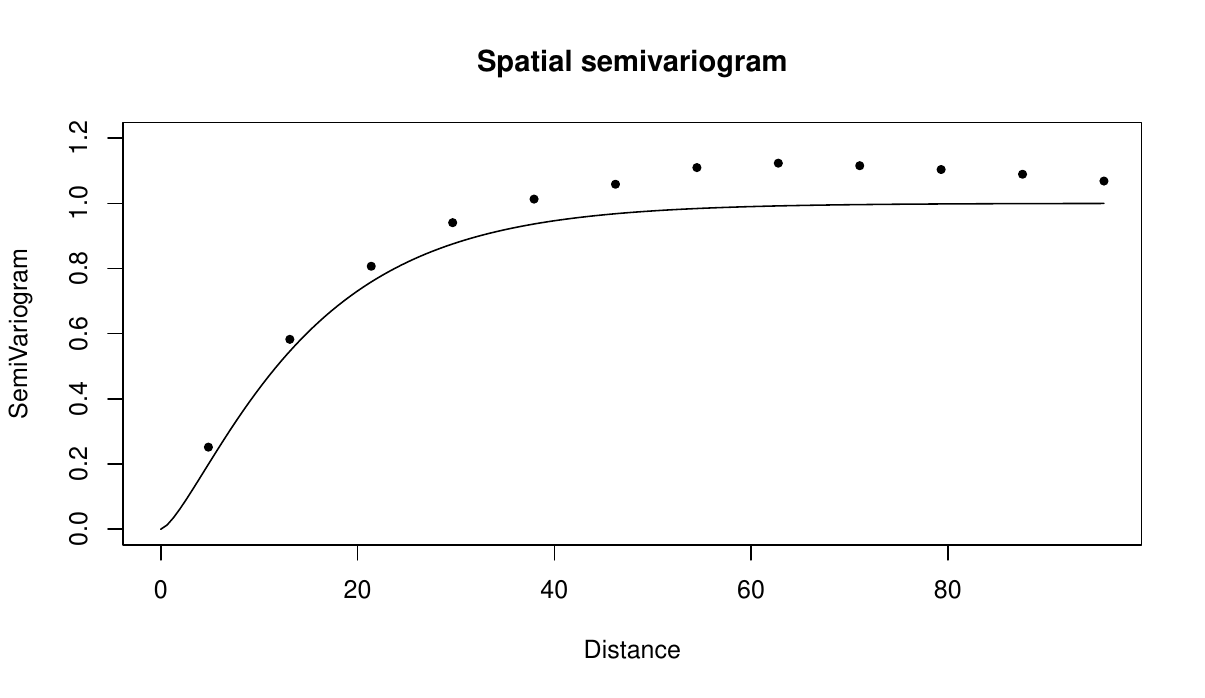}
\caption{Temperature application. Empirical semivariogram of the GAM residuals
and fitted Mat\'ern semivariogram for the Bernoulli stochastic NN pairwise fit
with \(K_{\mathrm{tar}}/n_{\mathrm{train}}=1\), \(m=10\), and \(p=0.10\).}
\label{fig:app_variogram}
\end{figure}

All stochastic NN pairwise fits produce nearly identical predictive accuracy.
Across the eight stochastic configurations, MAE ranges from 0.10237 to 0.10247,
RMSPE ranges from 0.18261 to 0.18293, and CRPS ranges from 0.08789 to 0.08870.
The Vecchia approximation gives the smallest MAE, RMSPE, and CRPS, but the
gain is very small: relative to the best stochastic configuration, the
reduction is about \(7.2\times10^{-5}\) in MAE and \(2.9\times10^{-4}\) in
RMSPE.

In the implementation-level comparison considered here, the main difference is
computational. The \texttt{GpGp} Vecchia covariance fit required about
3089 seconds, whereas the stochastic NN pairwise fits computed with
\texttt{GeoFit} required between about 126 and 159 seconds when
\(K_{\mathrm{tar}}/n_{\mathrm{train}}=1\), and between about 260 and
318 seconds when \(K_{\mathrm{tar}}/n_{\mathrm{train}}=2\). Thus, for the
specific workflows used here, the stochastic NN pairwise fits required roughly
10--25 times less covariance-fitting time than the Vecchia benchmark.

Overall, the application shows that stochastic NN pairwise likelihood can
deliver essentially the same predictive accuracy as the Vecchia benchmark in a
multi-million-observation setting. In the practical implementation-level
comparison considered here, this accuracy was achieved with substantially
shorter covariance-fitting times. In this large gridded temperature dataset,
the more aggressive budget \(K_{\mathrm{tar}}/n_{\mathrm{train}}=1\) was
already sufficient for prediction.

\section{Concluding remarks}
\label{sec:conclusion}

We have introduced a stochastic acceleration of nearest-neighbor weighted
pairwise composite likelihoods for massive spatial datasets. The proposed
approach separates the construction of a rich deterministic NN candidate graph
from the number of pairwise likelihood terms actually evaluated. Bernoulli
thinning controls this number in expectation, whereas fixed-budget thinning
retains an exact number of pairs by allocating the budget across target-specific
NN lists and sampling without replacement within each list.

The simulation studies suggest that retaining two pairs per observation provides
a stable statistical--computational compromise, in the Mat\'ern settings
considered, when all covariance parameters, including smoothness, are estimated.
The largest loss of efficiency is concentrated on smoothness estimation, while
the mean, scale, and sill parameters remain close to the deterministic NN
pairwise benchmark. The comparison with Vecchia shows that stochastic NN
pairwise likelihood can provide substantial reductions in covariance-fitting
time when a moderate loss of efficiency is acceptable. The internal-variability
study shows that the additional Monte Carlo variability due to thinning is
negligible for the mean and sill, moderate for the scale parameter, and largest
for smoothness.

The empirical budget \(K_{\mathrm{tar}}/n=2\) was calibrated using uniformly
distributed sampling locations on a two-dimensional domain and stationary
Mat\'ern covariance models. Other sampling designs, covariance families, or
forms of nonstationarity may require different retained-pair budgets. In
particular, strongly non-uniform sampling patterns may benefit from adaptive or
locally varying thinning probabilities, and covariance models with different
short-range behavior may require a different balance between the candidate graph
size \(m\) and the thinning parameter \(p\).

In the WorldClim temperature application, the stochastic NN pairwise fits give
predictive scores essentially indistinguishable from the Vecchia benchmark, but
with substantially shorter covariance-fitting times. The application also shows
that the proposed fits can be accompanied by Godambe standard errors computed by
a parametric score bootstrap, without refitting the model for each bootstrap
dataset.

Overall, the proposed stochastic NN pairwise likelihood provides a scalable
alternative for covariance estimation in massive spatial datasets. It is
particularly appealing when the main inferential or predictive targets are
mean, scale, sill, or spatial prediction, and when some loss of efficiency for
smoothness estimation is acceptable in exchange for substantial computational
savings.

Although this paper has focused on Gaussian spatial random fields, the proposed
stochastic thinning strategy only requires a candidate set of pairwise
likelihood contributions. It can therefore be adapted to non-Gaussian random
fields, for which pairwise composite likelihoods are often especially useful
\citep{Heagerty:Lele:1998,Bevilacqua_et_al:2021,poi22}, and to space--time or
multivariate random fields \citep{ppolow,gkleiber}, where the candidate graph
can be constructed from suitable nearest-neighbor pair lists. These extensions
are supported in the \texttt{GeoModels} implementation \citep{GeoModels} and
provide natural directions for future applied work.

\appendix
\section{Asymptotic properties under stochastic thinning}
\label{app:asymptotics}

This appendix states the asymptotic properties of the stochastic NN pairwise
estimators. The results extend the standard increasing-domain theory for NN
weighted pairwise composite likelihoods to the case where the NN weights are
multiplied by thinning indicators generated independently of the observed
field. The purpose is to isolate the additional role of the stochastic thinning
design, rather than to reproduce the full proof of the general weighted
pairwise likelihood theory.

Consistency, asymptotic normality, and Godambe covariance formulas for local
weighted pairwise likelihood estimators of spatial Gaussian random fields have
already been established under increasing-domain weak-dependence conditions;
see, for example, \citet{Bevilacqua:Gaetan:2015}. Conditional on the realized
thinning indicators, the stochastic criterion considered here is again a local
weighted pairwise likelihood, with random but design-measurable weights. For
notational simplicity, and consistently with the main text, the appendix is
written for the bivariate marginal pairwise contribution
\(\ell(Z_i,Z_j;\theta)\). The same arguments apply to other local pairwise
composite likelihood contributions, provided the corresponding regularity
conditions hold.

Let
\[
\mathcal E_n(m)=\{(i,j):w^{NN}_{ij}(m)=1,\ i\ne j\}
\]
be the deterministic directed NN candidate graph, and let
\(d_n=|\mathcal E_n(m)|\). The stochastic pairwise criterion is
\[
\widetilde{wpl}_{n}(\theta)
=
\sum_{(i,j)\in\mathcal E_n(m)}
X_{ij}\ell(Z_i,Z_j;\theta),
\]
where \(X_{ij}\in\{0,1\}\) are thinning indicators independent of the random
field. Equivalently, the stochastic weights are
\[
\widetilde w_{ij}=w^{NN}_{ij}(m)X_{ij}.
\]
Multiplication of \(\widetilde{wpl}_{n}(\theta)\) by a positive deterministic
normalizing constant, such as \((pd_n)^{-1}\) for Bernoulli thinning or
\(K_{\mathrm{tar}}^{-1}\) for fixed-budget thinning, does not change the
maximizer. It is nevertheless useful for stating limiting results.

Define the stochastic composite score and sensitivity matrix as
\[
\widetilde U_{n}(\theta)
=
\sum_{(i,j)\in\mathcal E_n(m)}
X_{ij}s_{ij}(\theta),
\qquad
s_{ij}(\theta)
=
\nabla_\theta \ell(Z_i,Z_j;\theta),
\]
and
\[
\widetilde H_{n}(\theta)
=
\sum_{(i,j)\in\mathcal E_n(m)}
X_{ij}h_{ij}(\theta),
\qquad
h_{ij}(\theta)
=
-\nabla^2_\theta \ell(Z_i,Z_j;\theta).
\]
Let
\[
K_n=\sum_{(i,j)\in\mathcal E_n(m)}X_{ij}
\]
be the number of retained pairs. For Bernoulli thinning,
\(K_n\) is random and \(E_X(K_n)=pd_n\). For fixed-budget thinning,
\(K_n=K_{\mathrm{tar}}\) exactly.

\begin{proposition}[Consistency and asymptotic normality]
\label{prop:asymptotics_thinning}
Assume that the spatial random field is correctly specified and satisfies the
standard increasing-domain weak-dependence and regularity conditions for local
weighted pairwise likelihood estimation. Assume also that the thinning design is
generated independently of the observed field, that the retained-pair size
satisfies \(K_n\to\infty\) and \(K_n=O(n)\), and that the normalized sensitivity
matrix is nonsingular in a neighborhood of the true parameter \(\theta_0\).
Then any sequence of maximizers \(\widehat\theta_n\) of the normalized
stochastic NN pairwise criterion is consistent for \(\theta_0\). Moreover,
\[
K_n^{1/2}
(\widehat\theta_n-\theta_0)
\Rightarrow
N\{0,H^{-1}JH^{-1}\},
\]
where
\[
H
=
\lim_{n\to\infty}
K_n^{-1}
E\{\widetilde H_{n}(\theta_0)\},
\qquad
J
=
\lim_{n\to\infty}
K_n^{-1}
\operatorname{Var}\{\widetilde U_{n}(\theta_0)\}.
\]
Here the expectation and variance are taken over both the random field and the
thinning design. Conditional versions of \(H\) and \(J\), given the realized
thinned edge set, are obtained by conditioning on the thinning indicators.
\end{proposition}

\begin{proof}[Sketch of proof]
Conditional on the thinning indicators, the stochastic criterion is a local
weighted pairwise composite likelihood with weights
\(\widetilde w_{ij}=w^{NN}_{ij}(m)X_{ij}\). The indicators are generated
independently of the observed field and depend only on the design, so the
selected weighted criterion remains a sum of local pairwise contributions.

Under the standard increasing-domain weak-dependence assumptions, the
normalized stochastic objective satisfies a uniform law of large numbers in a
neighborhood of \(\theta_0\). Since the pairwise margins are correctly
specified, the population criterion is maximized at \(\theta_0\), up to the
deterministic weighting induced by the NN graph and the thinning design.
Consistency follows from the usual argmax argument.

For asymptotic normality, expand the stochastic score around \(\theta_0\):
\[
0
=
\widetilde U_{n}(\widehat\theta_n)
=
\widetilde U_{n}(\theta_0)
-
\widetilde H_{n}(\bar\theta_n)
(\widehat\theta_n-\theta_0),
\]
where \(\bar\theta_n\) lies between \(\widehat\theta_n\) and \(\theta_0\).
A central limit theorem for local spatial sums gives
\[
K_n^{-1/2}\widetilde U_{n}(\theta_0)
\Rightarrow
N(0,J),
\]
and the normalized sensitivity satisfies
\[
K_n^{-1}\widetilde H_{n}(\bar\theta_n)
\overset{p}{\longrightarrow}
H.
\]
Combining these two limits gives the stated sandwich covariance
\(H^{-1}JH^{-1}\).
\end{proof}

For Bernoulli thinning, \(K_n\) is random. Under the retained-budget condition
\(pd_n\to\infty\), \(K_n/(pd_n)\to 1\) in probability, so normalizations based
on \(K_n\) and \(pd_n\) are asymptotically equivalent. For fixed-budget
thinning, \(K_n=K_{\mathrm{tar}}\) exactly. Since the retained-pair budgets used
in this paper are proportional to \(n\), the \(K_n^{1/2}\) and \(n^{1/2}\)
normalizations are equivalent up to a constant factor. The \(K_n^{1/2}\)
normalization is convenient here because the stochastic objectives are indexed
directly by the number of evaluated pairwise likelihood contributions.

The two thinning designs differ in how they affect the design variability. In
the constant-probability Bernoulli design,
\[
X_{ij}\stackrel{\mathrm{ind}}{\sim}\mathrm{Bernoulli}(p),
\qquad (i,j)\in\mathcal E_n(m),
\]
so that the expected thinned criterion is a positive multiple of the
deterministic NN pairwise criterion:
\[
E_X\{\widetilde{wpl}_{n}(\theta)\}
=
p
\sum_{(i,j)\in\mathcal E_n(m)}
\ell(Z_i,Z_j;\theta).
\]
Under correct specification, this multiplicative factor does not change the
population maximizer.

For fixed-budget thinning, \(K_n=K_{\mathrm{tar}}\) exactly and the sampling is
performed without replacement within target-specific NN lists. Conditional on
the local budget \(k_j\), any edge \(e=(i,j)\in\mathcal E_{n,j}(m)\) has
inclusion probability \(k_j/d_{n,j}\), and two distinct edges in the same list
are negatively correlated:
\[
\operatorname{Cov}_X(X_e,X_f\mid k_j)
=
-\frac{k_j(d_{n,j}-k_j)}
       {d_{n,j}^2(d_{n,j}-1)}
\le 0.
\]
Thus, fixed-budget thinning differs from Bernoulli thinning not only because it
fixes the retained size exactly, but also because it induces local negative
dependence among candidate pairs sharing the same target. This provides exact
computational control and a more regular allocation of retained pairs across
the NN graph, without implying uniform efficiency dominance over Bernoulli
thinning.

The Godambe covariance in Proposition~\ref{prop:asymptotics_thinning} is
defined under the unconditional-design interpretation, where variability is
taken over both the random field and the thinning design. This is the
interpretation used by the parametric score bootstrap in this paper, because
the thinning design is regenerated independently for each bootstrap dataset. A
conditional-design version is also possible: in that case, the realized thinned
edge set is kept fixed and \(J\) describes variability with respect to the
random field only.
\bibliographystyle{ECA_jasa}
\bibliography{references}

\end{document}